\documentclass[hidelinks,onefignum,onetabnum]{siamart251104}

\usepackage{braket,amsfonts,amsmath,amssymb}
\usepackage{array}
\usepackage[caption=false]{subfig}
\usepackage{algorithmic}
\usepackage{graphicx,epstopdf}
\usepackage{multirow}
\usepackage{cancel}
\usepackage{upgreek} 
\usepackage{mathrsfs} 
\usepackage{amsopn}
\usepackage{ulem}

\newsiamthm{claim}{Claim}
\newsiamremark{remark}{Remark}
\newsiamremark{hypothesis}{Hypothesis}
\crefname{hypothesis}{Hypothesis}{Hypotheses}
\newsiamthm{assumption}{Assumption}%[section]
\crefname{assumption}{Assumption}{Assumptions}
\crefname{ALC@unique}{Line}{Lines}

\newcommand{\cont}{\mathscr{C}}
\newcommand{\Cf}{\mathfrak{C}}

\newcommand{\R}{\mathbb{R}}
\newcommand{\C}{\mathbb{C}}
\newcommand{\Z}{\mathbb{Z}}
\newcommand{\N}{\mathbb{N}}

\newcommand\dd{{~\rm d}}
\newcommand\Hm{\mathcal{H}}
\newcommand\E{\mathbb{E}}
\newcommand\Var{\mathbb{V}}

\newcommand\rr{\mathbf{r}}

\newcommand{\D}{\mathcal{D}}

\newcommand{\uu}{\pmb{u}}
\newcommand{\F}{_{\rm{F}}}
\newcommand{\tr}{{\rm{Tr}}}

\newcommand{\dm}{P}
\newcommand{\rP}{\rho_\dm}

\newcommand{\dof}{n}

\newcommand{\Oc}{\mathcal{O}}
\newcommand{\PO}{M}
\newcommand{\PE}{p_{\PO}}

\newcommand{\ct}{{\rm{h}}}
\newcommand{\T}{{\rm{T}}}

\newcommand{\mS}{\mathbb{S}}

\newcommand{\pmp}{\upphi}

\newcommand{\ml}{\widehat{\mathcal{Q}}}

\newcommand{\mldm}{\widehat{\upphi}}
\newcommand{\mlc}{\mathcal{C}}
\newcommand{\mlv}{\mathcal{V}}
\newcommand{\im}{\mathrm{i}}

\newcommand{\f}{f}
\newcommand{\ff}{f^{\frac{1}{2}}}

\newcommand{\U}{[-1,1]}

\newcommand{\Ec}{E_{\rm c}}

\newcommand{\bR}{{\bar{R}}}
\newcommand{\Ecl}[1]{E_{\rm c}^{(#1)}}
\newcommand{\G}{\pmb{G}}
\newcommand{\LL}{\mathbb{L}}
\newcommand{\lmlsl}[2]{\mldm_{\chi^{(#1)}} ^{(#2)}}
\newcommand{\lmls}[1]{\mldm_{\chi} ^{(#1)}}
\newcommand{\chil}[1]{\chi^{(#1)}}
\newcommand{\ql}[1]{\ml^{(#1)}}
\newcommand{\nl}[1]{n^{(#1)}}
\newcommand{\Hl}{H^{(\ell)}}
\newcommand{\Hlm}{H^{(\ell-1)}}
\newcommand{\re}{{\rm Re}}
\newcommand{\PM}{\mathcal{P}_M}

\title{Stochastic Density Functional Theory Through the Lens of Multilevel Monte Carlo Method
}

\author{Xue Quan\thanks{Academy of Mathematics and Systems Science, Chinese Academy of Sciences, Beijing 100190, China (\email{xuequan@amss.ac.cn}).}
\and Huajie Chen \thanks{Corresponding. School of Mathematical Sciences, Beijing Normal University, Beijing 100875, China (\email{chen.huajie@bnu.edu.cn}).}}

\headers{SDFT through the Lens of Multilevel Monte Carlo}{Xue Quan and Huajie Chen}

\begin{document}
\maketitle

\begin{abstract}
The stochastic density functional theory (sDFT) has exhibited advantages over the standard Kohn--Sham DFT method and has become an attractive approach for large-scale electronic structure calculations.
The sDFT method avoids the expensive matrix diagonalization by introducing a set of random orbitals and approximating the density matrix via Chebyshev expansion of a matrix-valued function. 
In this work, we study the sDFT with a plane-wave discretization, and discuss variance reduction algorithms in the framework of multilevel Monte Carlo (MLMC) methods. 
In particular, we show that the density matrix evaluation in sDFT can be decomposed into many levels by increasing the plane-wave cutoffs or the Chebyshev polynomial orders. This decomposition renders the computational cost independent of the discretization size or temperature.
To demonstrate the efficiency of the algorithm, we provide a rigorous analysis of the statistical errors and present numerical experiments on some material systems.
\end{abstract}



\section{Introduction}
\label{sec:introduction}

Kohn--Sham density functional theory (DFT) \cite{Hohenberg64,Kohn65} is one of the most widely used electronic structure models in molecular simulations and materials science, which achieves an excellent balance between accuracy and computational cost. 
The cost of standard Kohn--Sham DFT scales cubically with system size because of repeated diagonalization of the Kohn--Sham Hamiltonian. 
To overcome this computational bottleneck, much effort has been devoted to the development of so-called linear-scaling DFT methods \cite{Baer1997,Goedecker1999,Kohn1996,Li1993,Mauri1993,ONETEP2020,Shimojo2008,Yang1991}.
However, most methods rely on Kohn's ``nearsightedness" principle \cite{Kohn1996} and are typically {only} applicable to non-metallic systems.

An alternative approach termed stochastic DFT (sDFT) \cite{Baer13,Fabian2019} has been introduced to accelerate the DFT calculations. 
The key idea is to approximate the density matrix by applying the Chebyshev polynomial expansion of the projection operator to a set of stochastic orbitals.
In particular, one approximates the density matrix by
$$
f(H)\approx p^2_M(H) = \sum_{i=1}^n \big(p_M(H) \pmb{\mathfrak{e}}_i\big)\big(p_M(H) \pmb{\mathfrak{e}}_i\big)^{\rm h}
\approx \sum_{\chi\in\mS} \big(p_M(H) \chi\big)\big(p_M(H) \chi\big)^{\rm h},
$$
where $H\in\C^{n\times n}$ denotes the discrete Hamiltonian matrix within certain discretization, $p_M$ is an $M$-th order Chebyshev approximation of $\ff$, $\{\pmb{\mathfrak{e}}_i\}$ is a complete basis set, and $\mS$ is a set of stochastic orbitals, where each $\chi\in\mS$ satisfies $\chi\in\C^n$ and $\E[\chi\chi^\ct]=I$.
The sDFT method circumvents the expensive orthogonalization procedure in standard Kohn–Sham DFT, reducing the computation to recursive
matrix-vector {products} for evaluating $p_M(H)\chi$.
{While nearsightedness provides the physical basis for many linear-scaling electronic-structure methods, the sDFT approach does not rely on an explicit locality or low-rank assumption on the density matrix. 
Instead, it uses stochastic orbitals to form unbiased estimators of density-matrix quantities, whose resulting accuracy is controlled by the estimator variance and the number of stochastic orbitals.
}

Since the stochastic sampling of the random orbitals introduces significant statistical noise, the efficiency of sDFT critically depends on effective variance reduction techniques.
Many strategies have been developed in recent years, including the embedded-fragmentation scheme \cite{Chen2019oefsdft,Neuhauser14}, the energy-window sDFT \cite{EWChenBaer2019}, and the tempering sDFT \cite{Nguyen2021}. 
{Another related approach is mixed stochastic-deterministic DFT, in which a selected set of deterministic Kohn--Sham orbitals is combined with stochastic orbitals to reduce the stochastic variance \cite{LiuChen2022,White20}.}
However, to our knowledge, there is no mathematical theory for variance reduction techniques in sDFT.
Recently, we have seen a first mathematical analysis for sDFT \cite{CaiLindsey2025}, which establishes nearly-optimal scaling in both the thermodynamic and complete basis set limits, by exploiting a discretization-independent stochastic estimator, a stochastic mirror descent reformulation of the self-consistent field iteration, and a pole expansion calculation.

The purpose of this work is to develop a variance reduction framework for sDFT and provide a rigorous analysis for the associated statistical errors.
We employ the multilevel Monte Carlo (MLMC) method, a powerful noise reduction strategy that uses a hierarchical structure: most samples are taken on computationally inexpensive coarse levels, while only a small number of samples are taken on the more expensive fine levels.
The MLMC method was first introduced in \cite{Giles2008,Giles2015}, and has since found widespread applications in numerous fields \cite{Anderson2012,Barth2011,Ullmann2015}, particularly in matrix trace estimation problems \cite{Frommer2022,Hallman2022} that are closely related to this work.  
We investigate the use of MLMC within sDFT and provide insights for existing variance reduction strategies. 
In particular, we study the sDFT with a plane-wave discretization and develop two distinct MLMC approaches for evaluating the Kohn–Sham map. 
These two approaches leverage hierarchical constructions in plane-wave cutoffs and Chebyshev expansion orders, making the cost independent of the discretization size and the temperature, respectively.

\vskip 0.1cm

\noindent
{\bf Outline.}
The rest of this work is organized as follows. 
In \Cref{sec:ksDFT}, we briefly review the Kohn--Sham DFT with a plane-wave discretization, and present a density matrix based formulation for the problem. 
In \Cref{sec:model}, we describe the sDFT approach and provide some basic variance analysis. 
In \Cref{sec:MLMC_sDFT}, we present a general MLMC framework to evaluate the Kohn--Sham map in sDFT, and investigate two multilevel decomposition strategies by exploiting the hierarchy of plane-wave discretization and polynomial expansion, respectively. 
In \Cref{sec:numerics}, we present numerical simulations on some material systems to demonstrate the efficiency of MLMC methods. 
In \Cref{sec:conclusion}, we provide some concluding remarks.

\vskip 0.1cm

\noindent
{\bf Notations.}
For $a\in\C$, we will denote the complex conjugate by $\overline{a}$.
For a matrix $A \in \C^{m \times n}$, we will denote its conjugate transpose by $A^\ct$.
For a random variable $X$, we will use $\E[X]$ and $\Var[X]$ to denote the expectation and variance, respectively. 
With slight abuse of notation, we will use $|\cdot|$ to denote the Euclidean norm of a vector, the volume of a bounded domain in $\R^d$, and the cardinality of a finite set. 
Throughout this paper, we will use the symbol $C$ to denote a generic positive constant that may change from one line to the next. The dependence of $C$ on model parameters will normally be clear from the context or stated explicitly.

\section{Kohn--Sham density functional theory}
\label{sec:ksDFT}
\setcounter{equation}{0}

In this section, we review the Kohn--Sham DFT with a plane-wave discretization.
In particular, the discrete Kohn--Sham equations are formulated as a fixed-point problem with respect to the density matrix, which provides the foundation of the sDFT method.

\subsection{Kohn--Sham equation}
\label{sec:kohn-sham}

In this work, we shall focus on systems with periodic boundary condition.
Let $d$ be the dimension, and $\Omega=A[0,1)^d$ be the domain with $A\in \R^{d\times d}$ a non-singular matrix. 
The associated Bravais lattice and reciprocal lattice are denoted by $\LL:=A\Z^d$ and $\LL^*:=2\pi A^{-\T}\Z^d$ respectively. 
{For $1\leq p\leq \infty$, we define 
$$
L^p_{\rm per}(\Omega) := \big\{ u\in L^p_{\rm loc}(\R^d):~u(\rr)=u(\rr+\tau)~{\rm for~any}~\tau\in\LL \big\} .
$$ 

We consider an $N$-electron system in this periodic cell and neglect spin variables.
Let $\rho:\Omega\rightarrow\R$ be the electron density satisfying $\int_\Omega\rho(\rr)\dd \rr =N$, the Kohn--Sham Hamiltonian associated with $\rho$ is given by
\begin{equation}
\label{KSham}
\Hm^{\rm KS}[\rho] : = -\frac{1}{2}\Delta + v_{\rm ps} + v_{\rm H}(\rho) + v_{\rm xc}(\rho) ,
\end{equation}
where the four terms denote, respectively, the kinetic operator, the external pseudopotential generated by the nuclei and core electrons, the Hartree potential describing the mean-field electrostatic interaction among electrons, and the exchange-correlation potential.}
{The pseudopotential is decomposed into a local part $v_{\rm loc}:\Omega\to\mathbb{R}$ and a finite-rank nonlocal part $v_{\rm nl}$
\begin{align*}
\big(v_{\rm nl}\phi\big)(\rr)=\sum_{j=1}^{m_{\rm nl}}(\xi_j, \phi)\xi_j(\rr)\qquad \forall ~\phi\in L^2_{\rm per}(\Omega) ,
\end{align*}
where $m_{\rm nl}\in \mathbb{N}_+$ denotes the rank of the nonlocal pseudopotential and $\xi_j\in L^2_{\rm per}(\Omega)$ are the corresponding projector functions. 
The Hartree potential can be written as 
$v_{\rm H}(\rho)(\rr):=\int_{\Omega}K_{\rm per}(\rr-\rr')(\rho(\rr')-N/|\Omega|)\dd\rr'$, with $K_{\rm per}$ denoting the periodic Coulomb kernel and $N/|\Omega|$ being the average electron density \cite{martin04}.}

{Throughout this paper, we assume that $v_{\rm loc}$, the projectors $\xi_j$, and the exchange-correlation potential $v_{\rm xc}$ are sufficiently smooth. 
In particular, their Fourier coefficients are assumed to decay exponentially fast.}
{We emphasize that the methodology and analysis developed in this work do not rely on a specific exchange-correlation functional. 
Numerical experiments for both LDA and PBE exchange-correlation models are presented in Section \ref{sec:numerics}.}

{Under the above conditions, $\Hm^{\rm KS}[\rho]$ in \cref{KSham} is a bounded self-adjoint perturbation of the periodic Laplacian.
Therefore, it is a self-adjoint operator on $L^2_{\rm per}(\Omega)$, bounded from below, and has compact resolvent.} 
The ground state of the system can {then} be obtained by solving the following Kohn--Sham equations
\begin{equation}
\label{KSequation}
\Hm^{\rm KS}[\rho]~\psi_i = \lambda_i\psi_i \qquad i=1,2,\cdots ,
\end{equation}
where $\int_{\Omega}\psi_i^* (\rr)\psi_j(\rr)\,{\rm{d}}\rr=\delta_{ij}$, and the electron density is given by the eigenpairs $\rho(\rr) = {\sum_{i=1}^\infty}f(\lambda_i)|\psi_i(\rr)|^2$ with some smearing function $f:\R\rightarrow [0,1]$.
{Here, the eigenvalues are listed in nondecreasing order and repeated according to their multiplicities, and the summation is taken over all eigenpairs with eigenvalues counted according to multiplicity.}

{Several smearing schemes are commonly used in electronic structure calculations, including Fermi--Dirac \cite{Dirac1926}, Gaussian \cite{Vita1991}, and Methfessel--Paxton smearing \cite{Methfessel1989}. 
In this work, we employ Fermi--Dirac smearing, which is consistent with the finite-temperature formulation and the associated density-matrix representation used in our analysis.} 
Specifically, the Fermi--Dirac distribution is defined as
\begin{equation}
\label{fermiDirac}
f(\lambda) = f_{\mu}(\lambda)
:= \frac{1}{1+\exp\big(\beta (\lambda-\mu)\big)} 
\qquad\forall~\lambda\in\R,
\end{equation}
where $\beta>0$ determines the smearing width, and $\mu\in\R$ is the Fermi level such that the condition ${\sum_{i=1}^\infty} f_{\mu}(\lambda_i)=N$ is satisfied.
Note that when $\beta\rightarrow\infty$, $f_{\mu}$ becomes a Heaviside function $1_{(-\infty,\mu)}$, in which case only the $N$ lowest eigenvalues of $H^{\rm KS}[\rho]$ need to be solved.

Note that the Kohn--Sham equation \cref{KSequation} is a nonlinear eigenvalue problem, as the operator $\Hm^{\rm KS}[\rho]$ depends on the eigenpairs $\big\{(\lambda_i,\psi_i)\big\}_{i\in\N_+}$ to be solved.
In practice, a self-consistent field (SCF) iteration is usually used to solve the nonlinear problem.
At each iteration step, the Hamiltonian is constructed with an input electron density, then an output electron density is obtained by solving the linearized eigenvalue problem and some charge-mixing technique \cite{Herbst_2021,martin04,Woods_2019}.
The iteration procedure is terminated when the difference between two consecutive densities is negligible.

\subsection{Plane-wave discretizations}
\label{sec:plane-wave}

In this work, we consider periodic boundary conditions for the Kohn--Sham equations and employ a plane-wave method for numerical discretization.

We denote by $e_{\G}({\rr})=|\Omega|^{-1/2} e^{\im\G\cdot\rr}$ the plane-wave basis with wavevector $\G\in\LL^*$. 
The family $\{e_{\G}\}_{\G\in\LL^*}$ forms an orthonormal basis set of $L_{\rm per}^2(\Omega)$.
Given an energy cut-off $\Ec>0$, we define the following finite dimensional subspace
\begin{equation}
\label{def:spaceE}
X_{\Ec}(\Omega) := \bigg\{u\in L_{\rm per}^2(\Omega) :~ 
u(\rr) = \sum_{\G\in\LL^*,~|\G|^2\leq 2\Ec} c_{\G}e_{\G}({\rr}) \bigg\} ,
\end{equation}
the dimension of which is denoted by $n:={\rm dim}\big(X_{\Ec}(\Omega)\big)=\big|\{\G\in\LL^*:~|\G|^2\leq 2\Ec\}\big|$.

With the plane-wave discretization, the eigenfunctions $\psi_i(\rr)$ are approximated by linear combinations of the basis functions $e_{\G}({\bf r})$ in $X_{\Ec}(\Omega)$, whose coefficients are denoted by $\uu_i:=\big\{\uu_{i,\G}\big\}_{|\G|^2\leq 2\Ec}\in\C^n$.
Then \cref{KSequation} is approximated by the following nonlinear algebraic eigenvalue problem
\begin{align}
\label{ks:discrete}
H[\rho_{\uu}] \uu_i = \varepsilon_i \uu_i \qquad i= 1,\cdots,n ,
\end{align}
where the electron density $\rho_{\uu}$ is associated to the approximate eigenfunction as
\begin{align}
\label{rho:discrete}
\rho_{\uu}(\rr) 
& = \sum_{i=1}^n \f(\varepsilon_i) \Bigg| \sum_{|\G|^2\leq 2\Ec} \uu_{i,\G}e_{\G}(\rr) \Bigg|^2
\\\nonumber
& = \sum_{|\G|^2\leq2\Ec}\sum_{|\G'|^2\leq2\Ec} \Big(\sum_{i=1}^n \f(\varepsilon_i)
{\overline{\uu_{i,\G}}}\uu_{i,\G'}\Big) \overline{e_{\G}(\rr)} e_{\G'}(\rr),
\end{align}
and the matrix elements of the discrete Hamiltonian $H[\rho_{\uu}]\in\C^{n\times n}$ are given by
\begin{align}
\label{HKS:discrete}
H[\rho_{\uu}]_{\G\G'} & = \int_{\Omega} \overline{e_{\G}(\rr)} \Hm^{\rm KS}[\rho_{\uu}] e_{\G'}(\rr) \dd\rr .
\end{align}
The {\it a priori} error estimates for the approximations of eigenvalue problems can be found in \cite{osborn75} for linear problems and in \cite{Cances3,cances12,chen13} for nonlinear Kohn--Sham equations.

\subsection{Kohn--Sham map of density matrix}
\label{sec:ksmap}

In the rest of this section, we shall rewrite the discrete Kohn--Sham equations \cref{ks:discrete} as a fixed-point problem with respect to the density matrix.
The density matrix formalism is fundamental for introducing the stochastic DFT methods.

Let $\D:=\big\{\dm\in\C^{\dof\times \dof}:~\dm=\dm^\ct\big\}$ be the space for Hermite matrices, endowed with the Frobenius inner product $(A,B)_{\F} :=\tr(A^\ct B)$.
The admissible set of discrete density matrices for an $N$-electron system is given by
\begin{equation}
\label{manifolds}
\D_N:=\Big\{\dm\in\D:~{0\preceq\dm\preceq 1},~\tr(\dm)=N\Big\}.
\end{equation}
{For any $P\in\D_N$, since $\dm$ is Hermitian, it admits a spectral decomposition.
More precisely, there exist $f_i\in [0,1]$ satisfying $\sum_{i=1}^n f_i=N$ and $\pmb{v}_i\in\C^n ~(i=1,\cdots,n)$ satisfying $\pmb{v}_i^{\rm h}\pmb{v}_j = \delta_{ij}$,
such that $P$ can be written in the form of
\begin{equation}
\label{dm:orbitals}
P = \sum_{i=1}^n f_i \pmb{v}_i\pmb{v}_i^{\rm h} .
\end{equation}
Given a density matrix $\dm\in\D_N$, we observe from \cref{rho:discrete} and \cref{dm:orbitals} that that the associated electron density to $\dm$ can be expressed by
\begin{equation}
\label{rho:P}
\rP(\rr) := \sum_{|\G|^2\leq2\Ec}\sum_{|\G'|^2\leq2\Ec} P_{\G\G'} \overline{e_{\G}(\rr)} e_{\G'}(\rr) .
\end{equation}
Although the choice of the orthonormal vectors $\pmb{v}_i$ is not unique, it does not affect $\dm$, and hence does not affect the associated electron density.}

With an ``input" density matrix $P\in\D_N$, we write the discrete Kohn--Sham Hamiltonian with a little abuse of notation
\begin{equation}
\label{HKS:dm}
H[\dm] := H[\rho_{\dm}] .
\end{equation}
Denoting the eigenpairs of $H[\dm]$ by $\big\{(\varepsilon_i,\uu_i)\big\}_{1\leq i\leq n}$, we have an ``output" density matrix $\sum_{i=1}^n \f(\varepsilon_i)\uu_i\uu_i^\ct$ by using \cref{dm:orbitals}.
We describe the map from an input density matrix to an output density matrix by the Kohn--Sham map $\Phi:\D_N\to\D_N$, with
\begin{equation}
\label{KSmap}
\Phi(\dm):= \f\big(H[\dm]\big) .
\end{equation}
Then the ground state density matrix ${\dm_*}$ is a fixed-point of the map $\Phi$, that is, it solves the the following equation
\begin{equation}
\label{dm:fix}
P_* = \Phi\big(P_*\big) .
\end{equation}
The simplest method for seeking the fixed point solution of \cref{dm:fix} is to use the iteration $\dm_{k+1} = \Phi(\dm_k)$, until $\Vert\dm_{k+1}-\dm_{k}\Vert_{\F}$ is sufficiently small. 
This iteration (together with some mixing scheme) is equivalent to performing SCF iterations to solve the Kohn--Sham equations.

The primary computational bottleneck in Kohn--Sham DFT is the evaluation of the Kohn--Sham map $\Phi(\dm)$. 
The direct evaluation requires a diagonalization of the Hamiltonian $H[\dm]$, whose cost scales cubically as $\Oc(n^3)$.
In the zero temperature limit, i.e., when $\beta\rightarrow\infty$ in \cref{fermiDirac}, only the lowest $N$ eigenpairs are required, and the cost also scales cubically as $\Oc(nN^2)$ due to the orthogonalization of $N$ eigenvectors.

\section{Stochastic DFT}
\label{sec:model}
\setcounter{equation}{0}

The stochastic DFT method cures the cubic scaling by employing two key strategies to reduce the cost of evaluating \cref{KSmap}: (i) random orbital approximation, and (ii) Chebyshev polynomial expansion.
In this section, we provide details of these approximations and present a basic error analysis.

Without loss of generality, we assume in the following that the spectrum of $H[\dm]$ lies in $[-1,1]$. 
This assumption is not essential as one can always apply a shift and rescaling to the discrete Hamiltonian by $a\big(H[\dm]-bI\big)$ (and also to the smearing function by $\f(\frac{\cdot}{a} +b)$) with appropriate constants $a$ and $b$.

\subsection{Random orbital approximation}
\label{sec:randorb}

We first introduce the notion of random orbital.
We call $\chi\in\C^\dof$ a random orbital if the entries of $\chi$ are independent random variables and satisfy the condition
\begin{equation}
\label{ass:randorb}
\E\big[\chi\chi^\ct\big] = I 
\end{equation}
with $I\in\R^{n\times n}$ the identity matrix.
Standard choices for $\chi$ satisfying \cref{ass:randorb} could be
(a) $\chi_j\in\{-1,1,-\im,\im\}$ with equal probability $\frac{1}{4}$; (b) $\chi_j=\exp (\im\theta)$ with $\theta$ uniformly distributed in $[0,2\pi]$; and (c) $\chi_j$ follows the normal distribution.

Let $\dm\in\D_N$ and $\{\pmb{\mathfrak{e}}_i\}$ be the standard basis set of $\R^n$, with $\pmb{\mathfrak{e}}_i$ having 1 at the $i$-th entry and 0 otherwise. 
By using \cref{KSmap} and \cref{ass:randorb}, we can write the Kohn--Sham map of the density matrix as an expectation in the following form
\begin{align}
\nonumber
& \Phi(\dm) = \f\big(H[P]\big) = \sum_{i=1}^n \f\big(H[P]\big) \pmb{\mathfrak{e}}_i \pmb{\mathfrak{e}}_i^\ct
\\[1ex] 
\label{dm:rand:1}
=~ & \E\Big(\f\big(H[P]\big)\chi\chi^\ct \Big)
\\[1ex] 
\label{dm:rand:2}
=~ & \E\Big[\pmp\big(\dm,\chi,\ff\big)\Big]
\quad{\rm with}\quad
\pmp(\dm,\chi,\ff) := \Big(\ff(H[\dm])\chi\Big)\Big(\ff(H[\dm])\chi\Big)^\ct .
\end{align}
{
Although \eqref{dm:rand:1} follows directly from 
$\mathbb E(\chi\chi^{\mathrm h})=I$, we keep the complete-basis representation to emphasize its stochastic interpretation: the deterministic summation over $\{\pmb{\mathfrak e}_i\}$ is replaced by an expectation over random orbitals.
This avoids applying $f(H[P])$ to all basis vectors and allows a reduction of the computational cost through stochastic approximation.
}
Since the Fermi--Dirac distribution satisfies $\f(x)\geq 0$, the matrix-valued function $\ff\big(H[P]\big)$ is well-defined.
Note that while equations \cref{dm:rand:1} and \cref{dm:rand:2} are mathematically equivalent, the form \cref{dm:rand:2} is preferred in sDFT calculations because its sample complexity is less sensitive to the discretization size \cite{CaiLindsey2025}. 

We can then approximate \cref{dm:rand:2} by employing a stochastic sampling approach.
Let $\mS$ be a set of independent random orbitals satisfying \cref{ass:randorb}, obtained through some practical sampling techniques.
Then the expectation in \cref{dm:rand:2} can be approximated by 
\begin{equation}
\label{DMapprox}
\Phi_{\mS}(\dm) := \frac{1}{|\mS|}\sum_{{\chi}\in\mS}\pmp\big(\dm,\chi,f^{\frac{1}{2}}\big).
\end{equation}
The accuracy of many quantities of interest is determined by the stochastic error of the density matrix approximation $\Phi_{\mS}(P)$. 
Markov's inequality \cite{markov} indicates that
\begin{align}
\label{estimate:dm:markov_ineq}
\mathbb{P} \bigg( \big\|\Phi_{\mS}(\dm)- \Phi(\dm) \big\|_{\rm F}<\Cf \sqrt{\E\Big[\big\|\Phi_{\mS}(\dm)- \Phi(\dm) \big\|_{\rm F}^2\Big]} \bigg) \geq 1- \Cf^{-2} 
\quad{\rm for}~ \Cf>0.
\end{align}
Using the facts that the orbitals $\chi\in\mS$ are independent and $\pmp(P,\chi,\ff)$ provides an unbiased estimate of $\Phi(\dm)$, we can express the second moment term in \cref{estimate:dm:markov_ineq} as
\begin{align}
\label{second_mom}
\E\Big[\big\|\Phi_{\mS}(\dm)- \Phi(\dm) \big\|_{\rm F}^2\Big]
& = \frac{1}{|\mS|}\Var\Big[ \pmp\big(P,\chi,\ff\big) \Big]
\qquad\qquad{\rm with}
\\[1ex]
\label{var:pmp}
\Var\Big[ \pmp\big(P,\chi,\ff\big) \Big] 
& := \E\left[ \Big\|\pmp\big(P,\chi,\ff\big)-\E\Big[\pmp\big(P,\chi,\ff\big)\Big]\Big\|_{\rm F}^2 \right] .
% = \sum_{i,j=1}^n\Var\Big[ \big(\pmp(P,\chi,\ff)\big)_{ij}\Big] 
\end{align}
We observe from \cref{estimate:dm:markov_ineq} and \cref{second_mom} that the accuracy of the stochastic approximation depends on (i) the number of random orbitals in $\mS$; and (ii) the variance of the random matrix.
Therefore, the variance \cref{var:pmp} is the key to the efficiency of the sDFT calculations. 
We provide an estimate for this variance in the following lemma, whose proof is given in \Cref{proof:lemma:var:sdft}.

\begin{lemma}
\label{lemma:var:sdft}
Let $\chi\in\C^n$ be the random orbital satisfying \cref{ass:randorb}.
Then
\begin{align}
\label{var1}
\nonumber
\Var\big[\pmp(\dm,\chi,\ff)\big]
& = \sum_{\substack{i,j=1\\i\neq j}}^\dof \big(\Phi(\dm)\big)_{ii}\big(\Phi(\dm)\big)_{jj}
+ \sum_{\substack{i,j=1\\i\neq j}}^\dof \big(\Phi(\dm)\big)_{ij}^2\E\big[\overline{\chi}_i^2\big]\E\big[\chi_j^2\big]
\\[1ex]
& \quad + \sum_{i=1}^\dof\big(\Phi(\dm)\big)_{ii}^2(\E\big[|{\chi}_i|^4\big]-1) .
\end{align}
Moreover, if $\chi$ additionally satisfies
\begin{align}
\label{chi:ass:2}
\E\big[\chi_i^2\big] = 0 
\quad{\rm and}\quad 
\E\big[|\chi_i|^4\big] = 1
\qquad{\rm for~any~} 1\leq i\leq n ,
\end{align}
then
\begin{equation}
\label{var2}
\Var\big[\pmp(\dm,\chi,f^{\frac{1}{2}})\big] =\sum_{\substack{i,j=1\\i\neq j}}^\dof\big(\Phi(\dm)\big)_{ii}\big(\Phi(\dm)\big)_{jj} 
= \big(\tr\left(\Phi(\dm)\right)\big)^2 - \sum_{i=1}^\dof\big(\Phi(\dm)\big)_{ii}^2 .
\end{equation}
\end{lemma}

{
\begin{remark}[Fourth-moment condition for random orbitals]
Among the three random orbital examples provided below equation \cref{ass:randorb}, both (a) and (b) satisfy the additional condition \eqref{chi:ass:2}, while (c) does not -- $\chi_i$ follows a normal distribution with the fourth moment $\E\big[|\chi_i|^4\big]=3$.
\end{remark}
}

\begin{remark}
[$L^2$-error of the electron density {under fixed-density scaling}]
\label{remark:err:rho}
Let $\dm\in\D_N$, and let $\rho_{\mS}$ and $\rho$ be the electron densities associated to $\Phi_{\mS}(\dm)$ and $\Phi(\dm)$ respectively.
By applying Markov's inequality \cite{markov}, the $L^2$-error of the electron density can be estimated by
\begin{align}
\label{estimate:rho:markov_ineq}
\mathbb{P} \left(\|\rho_{\mS}-\rho\|_{L^2(\Omega)}<\Cf \sqrt{\E\big[\|\rho_{\mS}-\rho\|^2_{L^2(\Omega)}\big]}\right) \geq 1- \Cf^{-2} 
\quad{\rm for}~ \Cf>0 .
\end{align}
It follows from \cref{rho:P} and the Cauchy–Schwarz inequality that
\begin{align}
\label{err:rho_L2}
\nonumber
\E\Big[\|\rho_{\mS}-\rho\|^2_{L^2(\Omega)}\Big] 
& = \frac{1}{|\Omega|}\sum_{\G} \E\bigg[ \bigg(\sum_{\G'} \Big(\big(\Phi_{\mS}(\dm)\big)_{\G',\G+\G'} -\big(\Phi(\dm)\big)_{\G',\G+\G'}\Big)\bigg)^2 \bigg]
\\[1ex]
& \leq\frac{1}{|\Omega|\cdot|\mS|}\sum_{\G} \bigg( \sum_{\G'}\sqrt{\Var\Big[ \big(\pmp(P,\chi,\ff)\big)_{\G',\G+\G'}\Big]} \bigg)^2 .
\end{align}
Since the variance of the random matrix element exhibits exponential decay (see e.g. \cref{var:proof:1}, \cref{var:proof:2} and \Cref{app:lemma:G_decay} in the appendix)
\begin{equation*}
\Var\Big[ \big(\pmp(P,\chi,\ff)\big)_{\G',\G+\G'}\Big] \leq \big(\Phi(\dm)\big)_{\G'\G'} \big(\Phi(\dm)\big)_{\G+\G',\G+\G'} \leq C e^{-\gamma(|\G'|+|\G+\G'|)} ,
\end{equation*}
we can apply Stechkin's lemma and derive the following estimate for the second moment term in \cref{err:rho_L2}
\begin{equation}
\label{err:rho:L2}
\E\Big[ \big\|\rho_{\mS} - \rho\big\|^2_{L^2(\Omega)} \Big] \leq \frac{C}{|\Omega|\cdot|\mS|}\Var\big[ \pmp(P,\chi,\ff)\big].
\end{equation}

{We assume $|\Omega|\propto N$, corresponding to a fixed-density supercell scaling. This is the standard setting for supercell convergence in crystalline systems, for example when larger systems are generated by repeating a primitive cell, possibly with structural defects (see \Cref{sec:numerics}).}
Since \cref{var2} and $\tr(\Phi(\dm)) = N$ imply that $\Var\big[\pmp(\dm,\chi,f^{\frac{1}{2}})\big] = \Oc(N^2)$, which together with \cref{estimate:rho:markov_ineq}, \cref{err:rho:L2}, and the scaling assumption $|\Omega|\propto N$ yield
$$
\|\rho_{\mS}-\rho\|_{L^2(\Omega)} = \Oc\Big(\sqrt{N/|\mS|}\Big) .
$$
Therefore, to achieve the same accuracy as the system size $N$ increases, the number of random orbitals in $\mS$ should be chosen proportional to $N$.
We mention that a different scaling for $|\mS|$ can be obtained when a different accuracy criterion is considered.
In many physics contexts \cite{Baer13,EWChenBaer2019}, the error averaged per electron is used as the metric, under which the number of required random orbitals can be significantly smaller.

{Note that the fixed-density scaling is not intended to cover all supercell limits. 
For instance, in slab or surface calculations, one may increase only the vacuum region while keeping the physical slab fixed. 
In this case, $|\Omega|$ may grow independently of $N$, and the assumption $|\Omega|\propto N$ no longer holds.
}
\end{remark}

\subsection{Chebyshev polynomial expansion}
\label{sec:chebyshev}

To evaluate the random matrix in \cref{dm:rand:2}, one needs to apply the matrix-valued function $\ff\big(H[P]\big)$ on each random orbital $\chi\in\mS$.
To avoid direct diagonalization of the Hamiltonian $H[P]$, $\ff$ is approximated on $[-1,1]$ by an $M$-th order polynomial expansion:
\begin{equation}
\label{CPexpansion}
\ff(x) \approx \PE(x) := \sum_{m=0}^{\PO}c_mT_m(x)
\qquad x\in [-1,1],
\end{equation}
where $T_m(\cdot)$ is the $m$-th Chebyshev polynomial and $c_m$ is the corresponding coefficient.
Then \cref{dm:rand:2} and \cref{DMapprox} can be further approximated by
\begin{align}
\label{reDMPC}
\Phi_{\PO}(\dm) := p_M^2\big(H[\dm]\big) & = \E\Big[ \pmp\big(\dm,\chi,\PE\big) \Big]
\\[1ex]
\label{reDMPCS}
& \approx \frac{1}{|\mS|}\sum_{{\chi}\in\mS} \pmp\big(\dm,\chi,\PE\big) =: \Phi_{\mS,M}(\dm) .
\end{align}
For each random orbital $\chi\in\mS$, the product $\PE\big(H[\dm]\big)\chi$ can be efficiently computed through recursive matrix-vector multiplications using Chebyshev polynomial recurrence relations. 
In practice, these matrix-vector multiplications are carried out using the Fast Fourier Transform (FFT), which allows the computational cost to scale {quasi}linearly with the matrix size $n$, {that is, $\Oc(n\log n)$}.

Note that $\ff$ is analytic on $\U$ but possesses singularities in the complex plane.
The locations of these singularities depend on both the smearing parameter $\beta$ and the Hamiltonian scaling constant $a$.
Consequently, the Chebyshev polynomial approximation error decays exponentially with the polynomial order $M$ \cite{Trefethen2013}
\begin{equation}
\label{cheberror}
\big\Vert\PE-\ff\big\Vert_{L^\infty(\U)} \leq C \exp\big(-\alpha \PO\big) ,
\end{equation}
where {the decay rate $\alpha\approx {\rm sinh}^{-1}(a\pi/\beta)$ is determined by the smearing parameter $\beta$ and the scaling constant $a$. Thus, a larger spectral range (a smaller $a$), or a lower electronic temperature (a larger $\beta$), decreases $\alpha$ and requires a higher Chebyshev polynomial degree for the same approximation accuracy.}

We now obtain an efficient approximation of the Kohn--Sham map $\Phi(\dm)$ for a given density matrix $\dm\in\D_N$, as $\Phi_{\mS,M}(\dm)$ defined in \cref{reDMPCS}.
This is the essence of sDFT, that avoids the expensive eigensolver by replacing the complete basis set with stochastic orbitals and approximating the smearing function with Chebyshev polynomial expansion. 
The approximation error of $\Phi_{\mS,M}$ is governed by the sampling size $|\mS|$, the variance estimated in \Cref{lemma:var:sdft}, and the polynomial order $M$.

\begin{remark}[Cost scaling]
\label{remark:cost}
The total computational cost for evaluating $\Phi_{\mS,M}(\dm)$ scales as {$\Oc(M|\mS|n\log n)$}.
Since $|\mS| \propto N$ (see \Cref{remark:err:rho}), the complexity becomes {$\Oc(MNn\log n)$}.
This cost depends on the system size $N$, the degrees of freedom $n$ (which is proportional to $N$ in most of the DFT calculations), and the smearing parameter $\beta$.
In the following section, we shall introduce two variance reduction strategies: one that makes the computational cost independent of $n$, and the other that makes it independent of $\beta$.
\end{remark}

\section{Multilevel Monte Carlo methods for sDFT}
\label{sec:MLMC_sDFT}
\setcounter{equation}{0}
 
In this section, we shall first present a general multilevel Monte Carlo (MLMC) framework to modify the standard sDFT scheme for evaluating the Kohn--Sham map in \cref{reDMPCS}, 
\begin{align*}
\Phi_{\mS,M}(\dm) = \frac{1}{|\mS|}\sum_{{\chi}\in\mS} \pmp(\dm,\chi,\PE) 
= \frac{1}{|\mS|}\sum_{{\chi}\in\mS} \Big(\PE\big(H[\dm]\big)\chi\Big)\Big(\PE\big(H[\dm]\big)\chi\Big)^\ct ,
\end{align*}
and then investigate two strategies of the multilevel decomposition, by exploiting the hierarchy of plane-wave discretization and polynomial expansion respectively.

\subsection{General framework}
\label{sce:MLMCframe}

To evaluate $\Phi_{\mS,M}$ within the MLMC framework, we shall first construct a sequence of approximations of $\PE\big(H[\dm]\big)$, denoted by
$$
\ml^{(\ell)}\in\C^{n\times n}, 
\qquad \ell=0,1,\cdots,L .
$$
Specifically, we set $\ml^{(L)}:=\PE(H[\dm])$ for the highest level.
As $\ell$ increases, the approximations $\ml^{(\ell)}$ exhibit a hierarchical structure, with increasing accuracy but also increasing computational cost.
This yields a sequence of increasingly precise approximations for the random matrix $\pmp(\dm,\chi,p_M)$ in \cref{reDMPCS} as
\begin{equation}
\label{sdm_l}
\mldm_\chi^{(\ell)} := \big(\ml^{(\ell)}\chi\big)\big(\ml^{(\ell)}\chi\big)^\ct, 
\qquad \ell=0,1,\cdots,L .
\end{equation}
We then obtain a multilevel decomposition of $\Phi_{M}(\dm)$ as
\begin{equation}
\label{pd_ml}
\Phi_{M}(\dm) = \E_{\chi^{(0)}}\big[\lmlsl{0}{0}\big] + \sum_{\ell=1}^L\E_{\chil{\ell}}\big[\lmlsl{\ell}{\ell}-\lmlsl{\ell}{\ell-1}\big] ,
\end{equation}
which gives the fundamental formula for the MLMC method.
At each level $\ell~(\ell=0,1,\cdots,L)$, we construct a distinct set of random orbitals $\mS^{(\ell)}$ to approximate the expectation value. The size of $\mS^{(\ell)}$ varies across levels to account for the different variances of the random matrices. 
After sampling the sets of random orbitals for all levels, we can approximate $\Phi_{M}(\dm)$ in \cref{pd_ml} with
\begin{equation}
\label{mlmc}
\Phi_{\{\mS^{(\ell)}\},M}(\dm) :=
\frac{1}{|\mS^{(0)}|}\sum_{\chi^{(0)}\in\mS^{(0)}}\lmlsl{0}{0} + \sum_{\ell=1}^L\bigg(\frac{1}{|\mS^{(\ell)}|}\sum_{\chi^{(\ell)}\in\mS^{(\ell)}}\Big(\lmlsl{\ell}{\ell}-\lmlsl{\ell}{\ell-1}\Big)\bigg).
\end{equation}
The multilevel strategy \cref{mlmc} outperforms the single level approach \cref{reDMPCS} by taking most of the samples at the computationally inexpensive low levels, and only a few samples at the costly high levels.
The efficiency of the MLMC methods in sDFT lies in two key aspects:
(i) Design appropriate hierarchical approximations $\big\{\ml^{(\ell)}\big\}_{\ell=0}^L$ for $\PE\big(H[\dm]\big)$ so that the variance decreases as $\ell$ increases; and (ii) Select optimal sampling sizes $|\mS^{(\ell)}|$ to balance the computational cost and the approximation error.
In the remainder of this subsection, we present the variance estimates and the strategy for optimizing the computational cost within a general MLMC framework.

With a hierarchical approximation $\big\{\ml^{(\ell)}\big\}_{\ell=0}^L$, the variance at each level is
\begin{align}
\label{mlvell}
\mlv_\ell := \Var[\lmls{\ell}-\lmls{\ell-1}]
\qquad {{\rm for}~\ell=1,\cdots,L}
\end{align}
{and $\mlv_0=\Var[\lmls{0}]$.}
The following lemma provides a general estimate of the variance $\mlv_\ell$ for $\ell=1,\cdots,L$, whose proof is given in \Cref{proof:lemma:ml_var}.

\begin{lemma}
\label{lemma:ml_var}
Let $\chi\in\C^n$ be the random orbital satisfying \cref{ass:randorb} and \cref{chi:ass:2}.
Then there exists a constant $C$ independent of $n,N,M$, such that 
\begin{align}
\label{var_sub}
\mlv_\ell \leq CN \big\| \ml^{(\ell-1)} - \PE(H[\dm]) \big\|_{\F}^2
\qquad \forall~1\leq\ell\leq L .
\end{align}
\end{lemma}

Intuitively, the estimate \cref{var_sub} shows that as $\ell$ increases, if $\ml^{(\ell)}$ provides a better approximation of $\PE(H[\dm])$, then the variance decays and hence fewer random orbitals are required at higher levels.

Then we shall investigate the design of an optimal sampling strategy across all levels.
{For each level $\ell$ and each random orbital $\chi\in\mS^{(\ell)}$, let $\mlc_\ell$ denote the computational cost of evaluating the corresponding single-sample contribution.
This is obtained by applying $\ml^{(0)}$ to $\chi$ for $\ell=0$, and by applying both $\ml^{(\ell)}$ and $\ml^{(\ell-1)}$ to $\chi$ and taking their difference for $\ell\geq 1$.}
We see from \cref{mlmc} that the total cost and the total variance of the MLMC method are given respectively by $\mlc = \sum_{\ell=0}^L \big|\mS^{(\ell)}\big| \mlc_\ell$ and $\mlv=\sum_{\ell=0}^L \mlv_\ell/\big|\mS^{(\ell)}\big|$. 
{According to \Cref{remark:err:rho}, the sampling error is $\Oc(\sqrt{\mlv/N})$. 
We therefore impose the variance constraint $\mlv/N=\epsilon^2$ and minimize the total cost $\mathcal C$, which leads to the following choice of the number of stochastic orbitals at each level
\begin{equation}
\label{num_orb}
\big|\mS^{(\ell)}\big| = \left\lceil\epsilon^{-2} \frac{1}{N}\sqrt{\frac{\mlv_\ell}{\mlc_\ell}} \Bigg(\sum_{\ell=0}^L\sqrt{\mlv_\ell\mlc_\ell}\Bigg)\right\rceil .
\end{equation}
}
Then the total cost for an MLMC approximation of the Kohn--Sham map, i.e., evaluation of \cref{mlmc}, is estimated by
\begin{align}
\label{ml_cost}
\mlc ~{\approx}~ \epsilon^{-2} \frac{1}{N} \Bigg(\sum_{\ell=0}^L\sqrt{\mlv_\ell\mlc_\ell}\Bigg)^2 .
\end{align}
Note that both the variance $\mlv_\ell$ and the computational cost $\mlc_\ell$ depend on the hierarchical construction of $\ml^{(\ell)}$.
Therefore, we can optimize the hierarchical structure by minimizing the total cost in \cref{ml_cost}, giving rise to the problem 
\begin{equation} 
\label{opt_ml_select}
\min_{\{\ml^{(\ell)}\}} \sum_{\ell=0}^L\sqrt{\mlv_\ell\mlc_\ell}.
\end{equation}

In the rest of this section, we shall develop two MLMC approaches for sDFT calculations leveraging hierarchical structures derived from the plane-wave discretization and Chebyshev polynomial expansion, respectively.
Note that the number of levels $L$ is fixed in the following discussions and numerical experiments. Nevertheless, our theories and methods can be extended to accommodate a varying $L$.

\begin{remark}[Cost reduction by MLMC]
\label{remark:mlmc:save}
In our following hierarchical constructions, the variance $\mlv_\ell$ decreases more rapidly than the cost $\mlc_\ell$ increases as the level $\ell$ grows. 
As a result, the product $\mlv_\ell \mlc_\ell$ decreases with $\ell$, implying that the total cost of the MLMC approximation is dominated by $\mlv_0 \mlc_0$. 
Since the cost of standard sDFT can be approximated as $\mlc_L \mlv_0$, the MLMC approach reduces the cost of sDFT by a factor of $\mlc_0 / \mlc_L$ (see also \cite{Giles2015}).
\end{remark}

\subsection{Multilevel on energy cutoffs}
\label{sec:Ec}

In DFT calculations, the primary numerical parameter that determines the accuracy and computational cost is the energy cutoff $\Ec$ for the plane-wave basis set \cref{def:spaceE}.
We first exploit the energy cutoff parameter to construct hierarchical approximations of $p_M(H[\dm])$.

Let
$$
0<\Ec^{(0)}<\cdots<\Ec^{(L)}=\Ec
$$
be a sequence of energy cutoffs.
At the $\ell$-th level, the energy cutoff $\Ec^{(\ell)}$ defines a finite dimensional subspace $X_{\Ecl{\ell}}(\Omega)$, whose degrees of freedom is given by $\nl{\ell}:={\rm dim}\big(X_{\Ecl{\ell}}(\Omega)\big) = \mathcal{O}\big((\Ec^{(\ell)})^\frac{d}{2}\big)$. 
For each level $\ell = 0, \dots, L$, we construct the corresponding Hamiltonian matrix $\Hl [\dm] \in \C^{\nl{\ell} \times \nl{\ell}}$, whose matrix elements follow the same definition as \cref{HKS:discrete}.
The idea is to approximate $p_M(H[\dm])$ by $p_M\big(\Hl [\dm]\big)$ at the $\ell$-th level.
Since $\PE(\Hl [\dm])\in\C^{\nl{\ell} \times \nl{\ell}}$ has a smaller matrix size than the full one $p_M(H[\dm])\in\C^{n\times n}$, it must be expanded by zero-padding to ensure compatibility with the MLMC formula \cref{mlmc}.
In particular, matrix elements of $\ml^{(\ell)}\in \C^{n\times n}$ at the $\ell$-th level are defined by 
\begin{equation} 
\label{ec_h}
\big(\ml^{(\ell)}\big)_{\G\G'} := 
\left\{\begin{array}{ll} \PE\big(\Hl [\dm]\big)_{\G\G'} \qquad &{\rm if}~ |\G|^2\leq 2\Ec^{(\ell)} ~~ {\rm and} ~~ |\G'|^2\leq 2\Ec^{(\ell)},
\\[1ex]
0 & {\rm otherwise}.
\end{array}
\right. 
\end{equation}
The random matrix $\mldm_\chi^{(\ell)}$ in \cref{sdm_l} can be obtained accordingly.

{At level $\ell$, the evaluation of each single-sample contribution is performed by recursive matrix-vector products. 
Let $\widetilde{\mlc}(\Ec):=\Ec^{d/2}\log(1+\Ec^{d/2})$.
Then the corresponding per-orbital cost scales as
\begin{align}
\label{cl:plane-wave}
&\mlc_0\propto \nl{0} \log \nl{0} \propto \widetilde{\mlc}(\Ec^{(0)})
\qquad\qquad{\rm and}
\\[1ex]
\nonumber
&\mlc_\ell \propto \nl{\ell} +\nl{\ell-1}\propto \widetilde{\mlc}(\Ec^{(\ell)}) +\widetilde{\mlc}(\Ec^{(\ell-1)}),\quad \ell=1,\cdots,L.
\end{align}
}
We shall then use \Cref{lemma:ml_var} to estimate the variance $\mlv_\ell$ with the multilevel construction \cref{ec_h}.
The following result shows that the variance decays exponentially with respect to the increasing energy cutoffs, the proof of which is given in \Cref{proof:thm:mlmc_ec}.
 
\begin{theorem}
\label{thm:mlmc_ec}
Let $\chi$ be the random orbital satisfying \cref{ass:randorb} and \cref{chi:ass:2}, $\ml^{(\ell)}$ be given by \cref{ec_h}, $\mldm_\chi^{(\ell)}$ be given by \cref{sdm_l}, and $\mlv_\ell$ be given by \cref{mlvell}.
Then there exist constants $C$ and $\gamma$ depending on the smearing parameter $\beta$, such that 
\begin{equation} 
\label{ec:val_l}
\mlv_\ell \leq CN \left( \exp\big(-2\gamma \sqrt{\Ec^{(\ell-1)}}\big)
+ \exp\big(-2\alpha \PO)\right) 
\qquad \forall~\ell = 1,\cdots,L .
\end{equation}
\end{theorem}

The second term in the estimate \cref{ec:val_l} is not relevant to the hierarchical construction \cref{ec_h}, but comes from the polynomial approximation error.
In practice, we can take sufficiently large $M$ so that the variance becomes dominated by the first part of \cref{ec:val_l}.
{Note that according to \Cref{lemma:var:sdft}, we have $\mlv_0 = \Oc(N^2)$.} Combining these with the cost $\mlc_\ell$ in \cref{cl:plane-wave}, the optimization problem \cref{opt_ml_select} can be reformulated as the following problem that finds the optimal energy cutoffs across all levels
{
\begin{align}
\label{ec:opt}
\min_{\big\{\Ec^{(\ell)}\big\}} \bigg\{&\sqrt{N\widetilde{\mlc}(\Ec^{(0)})}+\sum_{\ell=1}^L \exp\Big( -{\gamma}\sqrt{\Ec^{(\ell-1)}} \Big)  \sqrt{\widetilde{\mlc}(\Ec^{(\ell)})+\widetilde{\mlc}(\Ec^{(\ell-1)})} \bigg\}.
\end{align}
In deriving \cref{ec:opt}, we have omitted multiplicative constants and common prefactors that are independent of the cutoff parameters, since they do not affect the minimizer. 
The remaining expression keeps the relevant dependence on the cutoffs and on the system size $N$.
The resulting optimization should be understood as an idealized asymptotic cost model. 
In an actual plane-wave implementation, the basis consists of the discrete set of reciprocal lattice vectors satisfying the cutoff condition.
Hence increasing the energy cutoff changes the number of plane waves only in a stepwise manner, rather than exactly according to the continuum scaling law. 
In practice, the cost should be interpreted in terms of the actual number of selected plane-waves at each level.}

For simplicity of implementations, we can choose energy cutoffs that follow an algebraic growth rate across levels 
\begin{equation}
\label{ecl:p}
\Ec^{(\ell)} = E_0 + (\Ec-E_0)\left(\frac{\ell+s}{L+s}\right)^p \quad{\rm with} ~~ p>0 .
\end{equation}
Here, $E_0>0$ is a small energy cutoff for level $\ell=0$ to capture the key characteristics of the system, and the parameter $s$ is selected empirically.
Then we can optimize the rate parameter $p$ in \cref{ecl:p} to determine the optimal cost for this MLMC strategy (see our numerical experiments in \Cref{sec:numerics}).

\begin{remark}
[Cost scaling for multilevel on energy cutoffs]
\label{remark:cost:mlec}
We observe that the second term in the objective function in \cref{ec:opt} can be bounded by a constant independent of the finest cutoff $\Ec$
{
\begin{align*}
\sum_{\ell=1}^L e^{-{\gamma}\sqrt{\Ec^{(\ell-1)}}} \sqrt{\widetilde{\mlc}(\Ec^{(\ell)})+\widetilde{\mlc}(\Ec^{(\ell-1)})}
\leq C\int_{\Ecl{0}}^{\Ec} e^{-\gamma\sqrt{x}}x^{\frac{d}{4}+\frac{1}{2}}\dd x \leq C\Gamma\bigg(\frac{d}{2}+3\bigg) ,
\end{align*}
where $\Gamma(\cdot)$ denotes the Gamma function.
Here, the first inequality follows by comparing the discrete sum with the corresponding integral and the bound $\sqrt{\log(1 + x^{d/2})} \leq C \sqrt{x}$.} 
As a result, the total cost of the MLMC scheme with multilevel energy cutoffs is dominated by the cost at level $\ell=0$, which is $\Oc({\epsilon^{-2}}N^{-1}\mlv_0\mlc_0)$. 
Together with the fact $\mlv_0=\Oc(N^2)$ and {$\mlc_0\propto Mn^{(0)}\log n^{(0)}$}, the total cost scales as {$\Oc(\epsilon^{-2}NMn^{(0)}\log n^{(0)})$}, which does not depend on the discretization size $n$ (and hence also independent of the finest energy cutoff $\Ec$).
\end{remark}

\subsection{Multilevel on polynomial orders}
\label{sec:ml:polynomial}

The polynomial order $M$ serves as another crucial numerical parameter that governs both accuracy and computational cost in sDFT. 
This motivates us to construct hierarchical approximations of $p_M(H[\dm])$ by selecting a sequence of increasing polynomial orders.

Let
$$
0<\PO^{(0)}<\dots<\PO^{(L)}=\PO.
$$
We approximate $p_M(H[\dm])$ by
\begin{equation} 
\label{spd_l}
\ql{\ell} := p_{\PO^{(\ell)}}\big(H[\dm]\big) ,
\end{equation}
which gives rise to the corresponding random matrix $\mldm_\chi^{(\ell)}$ in \cref{sdm_l}.
{By using the recursive matrix-vector products, the computational cost for the evaluation of each single-sample contribution scales as
\begin{equation} 
\label{cost:dgree}
\mlc_\ell\propto\PO^{(\ell)}, \qquad\ell=0,1,\cdots,L.
\end{equation}
}
We next analyze the variance $\mlv_\ell$ at each level. 
The following result establishes that the variance $\mlv_\ell$ decays exponentially as $\ell$ increases, whose proof is given in \Cref{proof:thm:mlmc_pd}.

\begin{theorem}
\label{thm:mlmc_pd}
Let $\chi$ be the random orbital satisfying \cref{ass:randorb} and \cref{chi:ass:2}, $\ml^{(\ell)}$ be given by \cref{spd_l}, $\mldm_\chi^{(\ell)}$ be given by \cref{sdm_l}, and $\mlv_\ell$ be given by \cref{mlvell}.
Then there exist constants $C$ and $\alpha$ depending on the smearing parameter $\beta$, such that 
\begin{equation} 
\label{pd:val_l}
\mlv_\ell \leq C N\exp\big( -2\alpha M^{(\ell-1)} \big) 
\qquad \forall~\ell=1,\cdots,L.
\end{equation}
\end{theorem}

Combining \cref{cost:dgree}, \cref{pd:val_l} and $\mlv_0=\Oc(N^2)$, we approximately obtain the optimization problem \cref{opt_ml_select} for polynomial orders
\begin{equation}
\label{pd:opt}
\min_{\{M^{(\ell)}\}} \bigg\{{\sqrt{NM^{(0)}}} + \sum_{\ell=1}^L \exp\Big( -\alpha M^{(\ell-1)} \Big) \sqrt{M^{(\ell)}}\bigg\} . 
\end{equation}
Analogous to the energy cutoff selection approach in \cref{ecl:p}, we assume an algebraic growth of $M^{(\ell)}$, by taking
\begin{equation}
\label{pdl:q}
M^{(\ell)} =M_0 + \left\lceil (M-M_0)\left(\frac{\ell+t}{L+t}\right)^q\right\rceil \quad{\rm with} ~~ q>0 .
\end{equation}
Here, $M_0\in\Z_+$ is a small polynomial order to roughly capture the characteristics of the function $f$, and the parameter $t$ is selected empirically.
We then optimize the rate parameter $q$ in \cref{pdl:q} to determine the optimal cost for this MLMC strategy (see our numerical experiments in \Cref{sec:numerics}).

\begin{remark}[Cost scaling for multilevel on polynomial orders]
\label{remark:cost:mlpd}
Similar to \Cref{remark:cost:mlec}, the cost of the multilevel on polynomial orders is $\Oc({\epsilon^{-2}}N^{-1}\mlv_0\mlc_0)$. Combined with $\mlv_0=\Oc(N^2)$ and {$\mlc_0\propto M^{(0)}n\log n$}, this multilevel approach yields a total computational cost of {$\Oc(\epsilon^{-2}NM^{(0)}n\log n)$}, which is independent of the polynomial order $M$ (and hence also independent of the smearing parameter $\beta$).
\end{remark}

\section{Numerical experiments}
\label{sec:numerics}

In this section, we present numerical simulations of several {two- and three-dimensional material systems} using sDFT. 
{The computations are implemented based on the packages DFTK.jl \cite{DFTKjcon} (a library of Julia routines for plane-wave DFT algorithms) and SDFT.jl \cite{SDFT} (our Julia package for sDFT simulations).
}

\vskip 0.2cm

\noindent
{\bf {Example 1 (2D systems).}}
We consider three graphene-type systems of increasing complexity: perfect graphene lattice, graphene with a Stone--Wales defect, and boron/nitrogen-doped graphene (with a doping ratio of 10\%). 
The corresponding atomic configurations are shown in \Cref{fig:graphene}. 
The system size $N$ is varied from 16 to 128 through the construction of supercells with different sizes.
{All simulations for 2D systems are carried out on a workstation equipped with an Intel Xeon W-3275M 56-core processor and 1.5 TB of RAM.}

{For the DFT models, we use PseudoDojo UPF norm-conserving pseudopotentials in separable (Kleinman--Bylander) form \cite{Kleinman1982,VANSETTEN201839},}
{together with the exchange-correlation functional consisting of Slater exchange \cite{Dirac_1930} and Perdew--Wang correlation \cite{perdew92}. 
In all simulations, only the $\Gamma$ point is used for $k$-point sampling in the Brillouin zone.}

We first test the variance of the random matrix $\pmp(\dm,\chi,f^{\frac{1}{2}})$ by theoretical prediction (from \Cref{lemma:var:sdft}) and numerical tests (from averaging 10 independent simulations).
We choose the energy cutoff at $\Ec=20$ Ha and the smearing parameter at $\beta = 10^3~{\rm Ha}^{-1}$, and show the growth of the variance with respect to the system size in \Cref{fig:var:dm}. 
The results demonstrate excellent agreement between the numerically sampled variance and the theoretical prediction from \cref{var2} across all three different types of systems. 
Moreover, we see that the variance exhibits a quadratic increase with respect to the system size. 

We further test the stochastic error of the electron density
$$
\Delta\rho := \|\rho_{\mS} - \rho_{\rm ref}\|_{L^2(\Omega)},
$$
where the reference density $\rho_{\rm ref}$ is obtained from a standard DFT calculation using the complete basis set $\{\mathfrak{e}_i\}$. 
In \Cref{fig:std:rho}, we compare the stochastic error for {supercells} of sizes $2\times 2$, $3\times 3$ and $4\times 4$. We observe that the error scales linearly with $(N/|\mS|)^{1/2}$ across all systems, which is consistent with our analysis in \Cref{remark:err:rho}.

\begin{figure}[htbp!]
\centering
    \subfloat[Perfect lattice]{\label{fig:perfect_lat}\includegraphics[height=3.7cm]{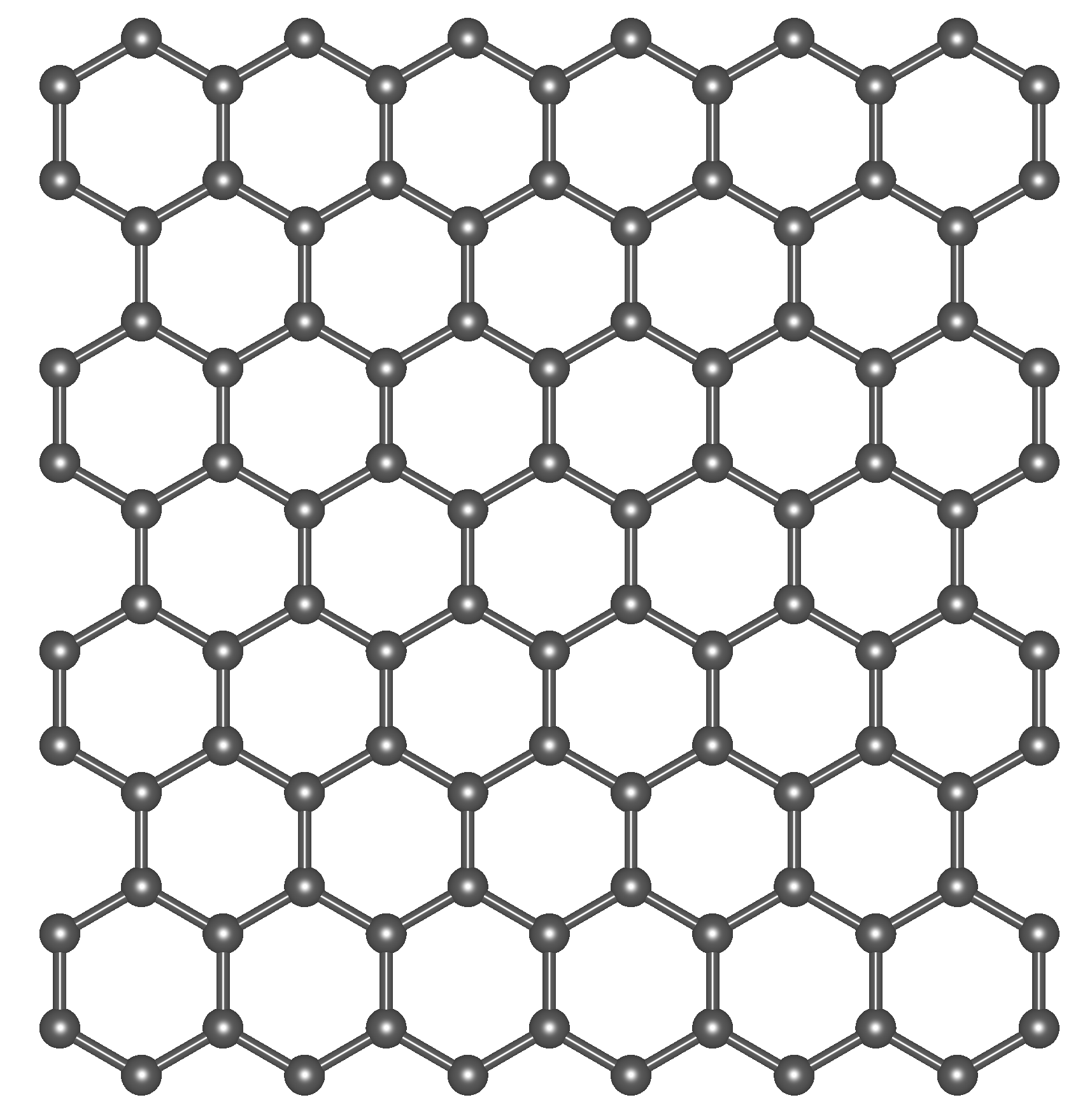}}
    \hskip 0.5cm
    \subfloat[Stone--Wales]{\label{fig:stone_wales}\includegraphics[height=3.7cm]{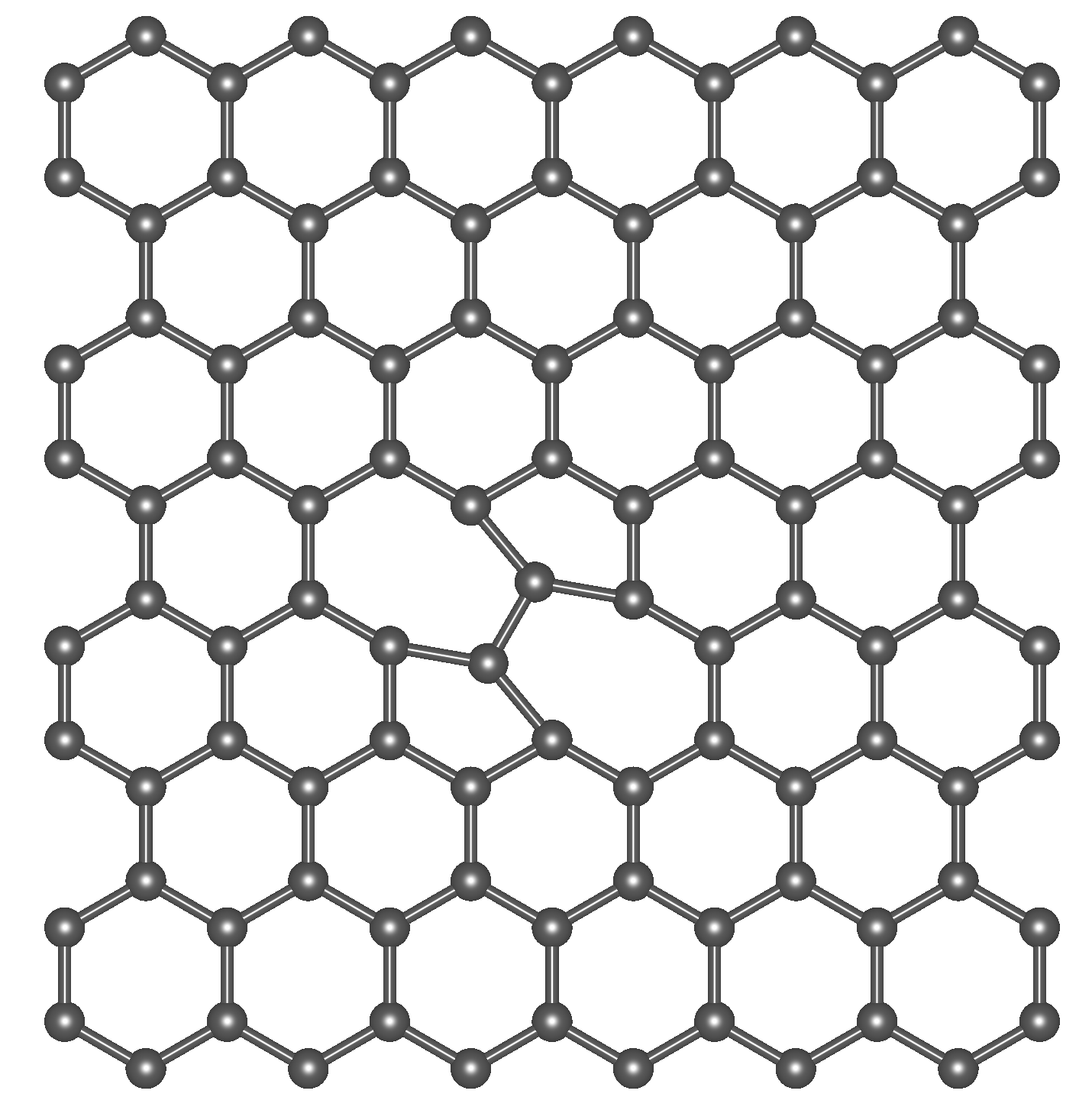}}
    \hskip 0.5cm
    \subfloat[Doping]{\label{fig:doping}\includegraphics[height=3.7cm]{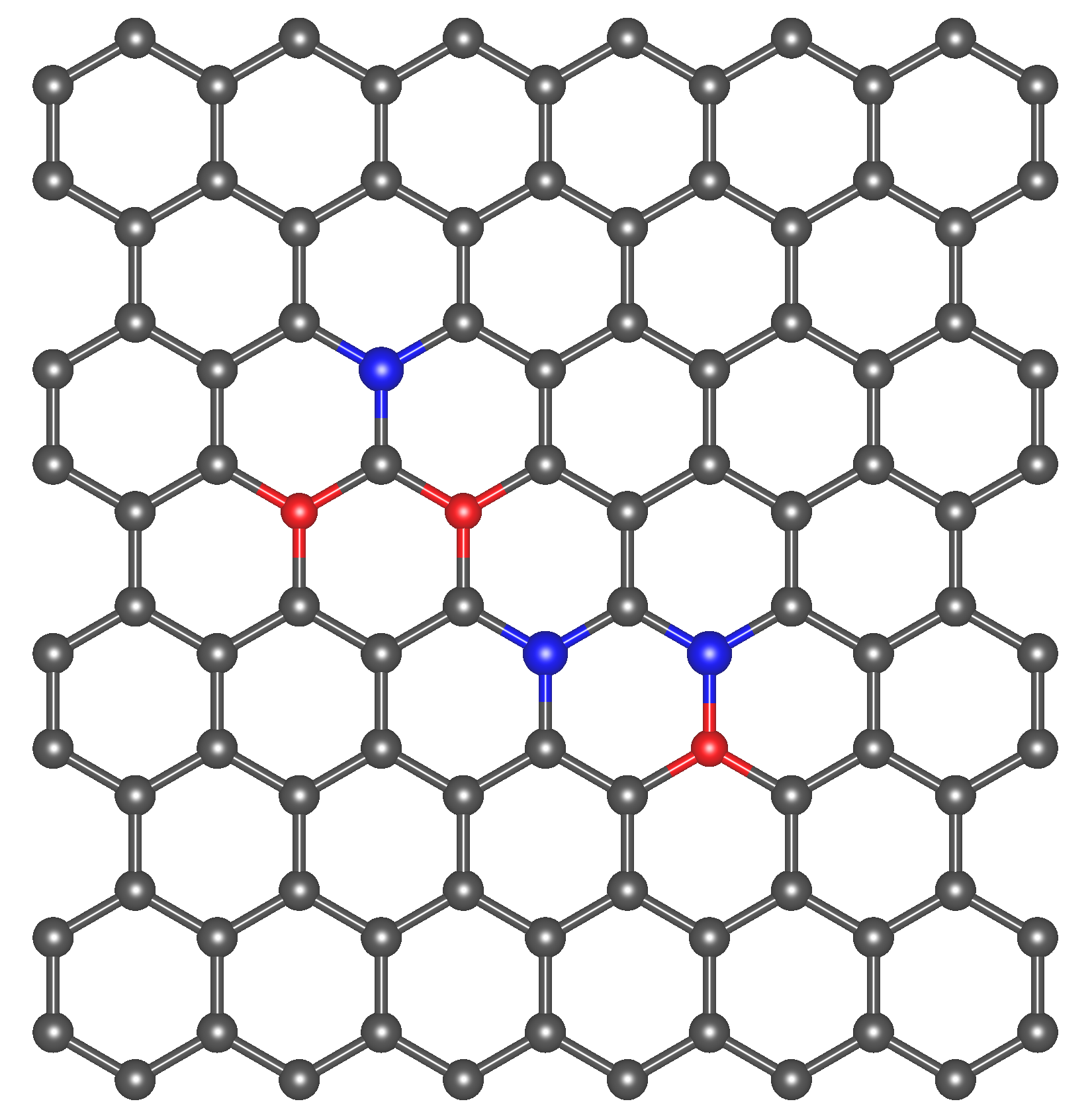}}
    \caption{Atomic configurations of the grahene-type systems. (a) Perfect lattice. (b) Stone--Wales defect. (c) Boron (blue) / nitrogen (red) doped graphene.} 
    \label{fig:graphene}
\vskip 0.2cm
\centering
    \subfloat[Perfect lattice]{\label{fig:sdft:var:graphene}\includegraphics[height=3.8cm]{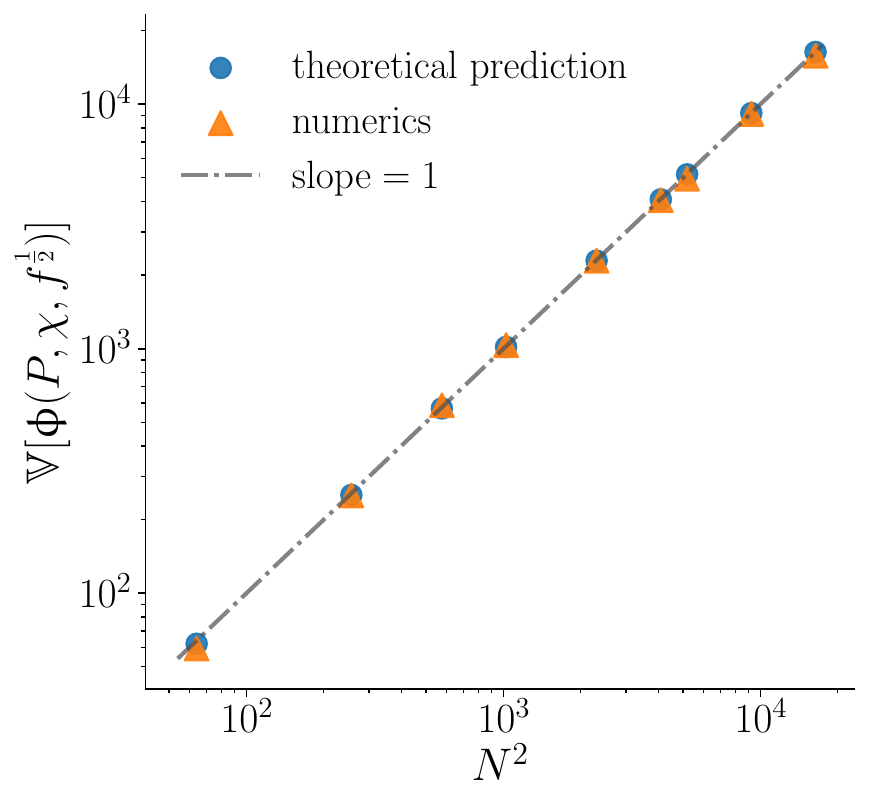}}
    \hskip 0.2cm
    \subfloat[Stone--Wales]{\label{fig:sdft:var:stone_wales}\includegraphics[height=3.8cm]{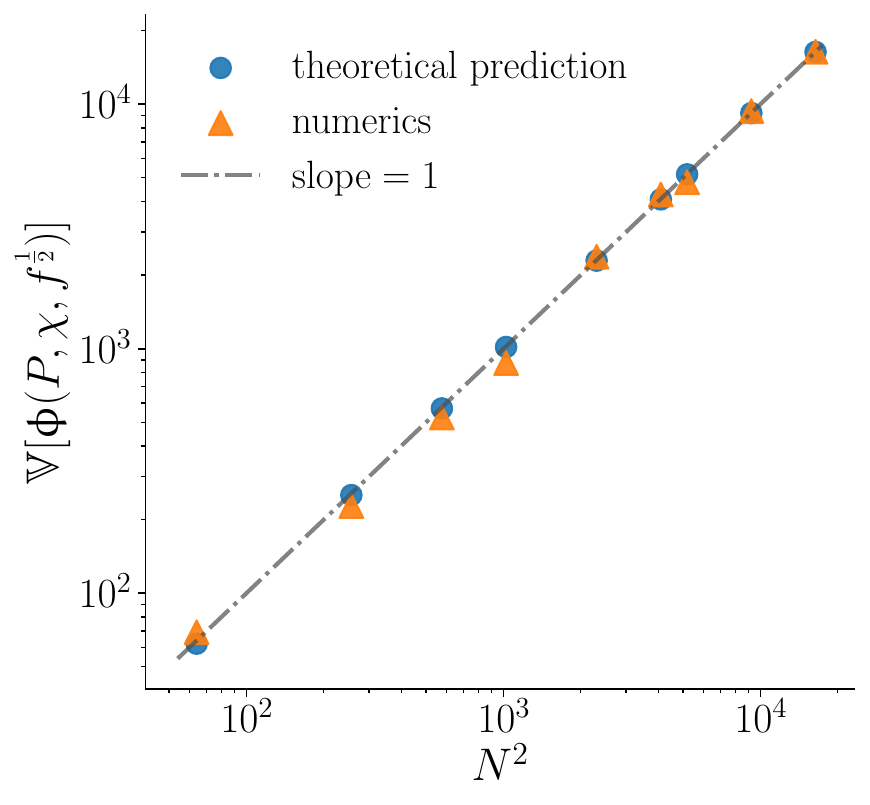}}
    \hskip 0.2cm
    \subfloat[Doping]{\label{fig:sdft:var:doping}\includegraphics[height=3.8cm]{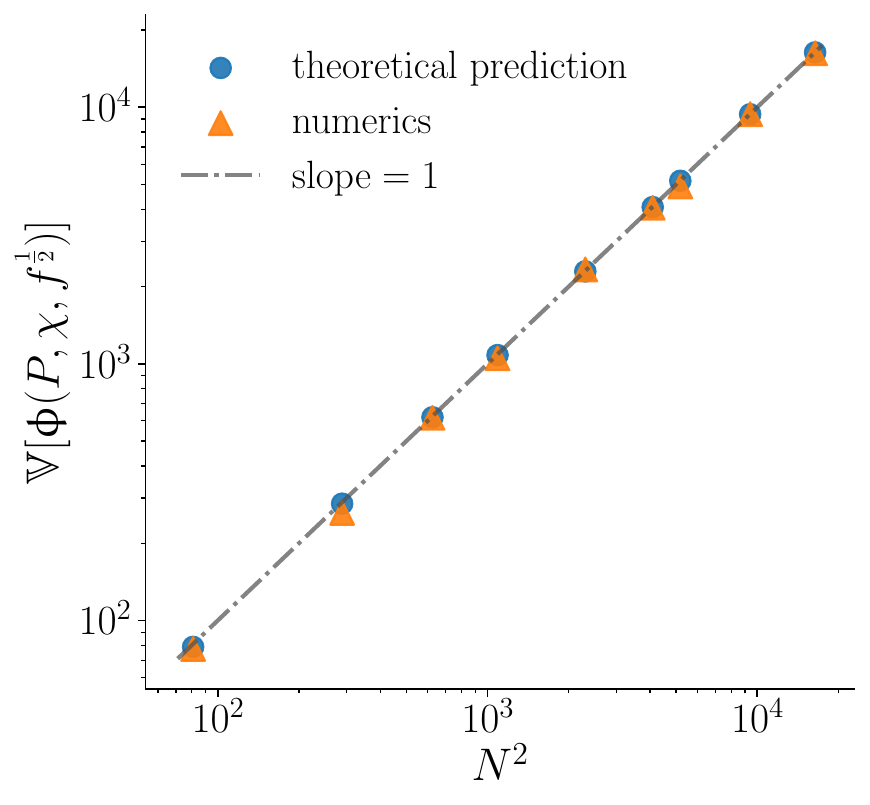}}
    \caption{The variance of the random matrix $\pmp(\dm,\chi,f^{\frac{1}{2}})$ with respect to the number of electrons.}
    \label{fig:var:dm}
\vskip 0.2cm
\centering
    \subfloat[Perfect lattice]{\label{fig:sdft:err:graphene}\includegraphics[height=3.8cm]{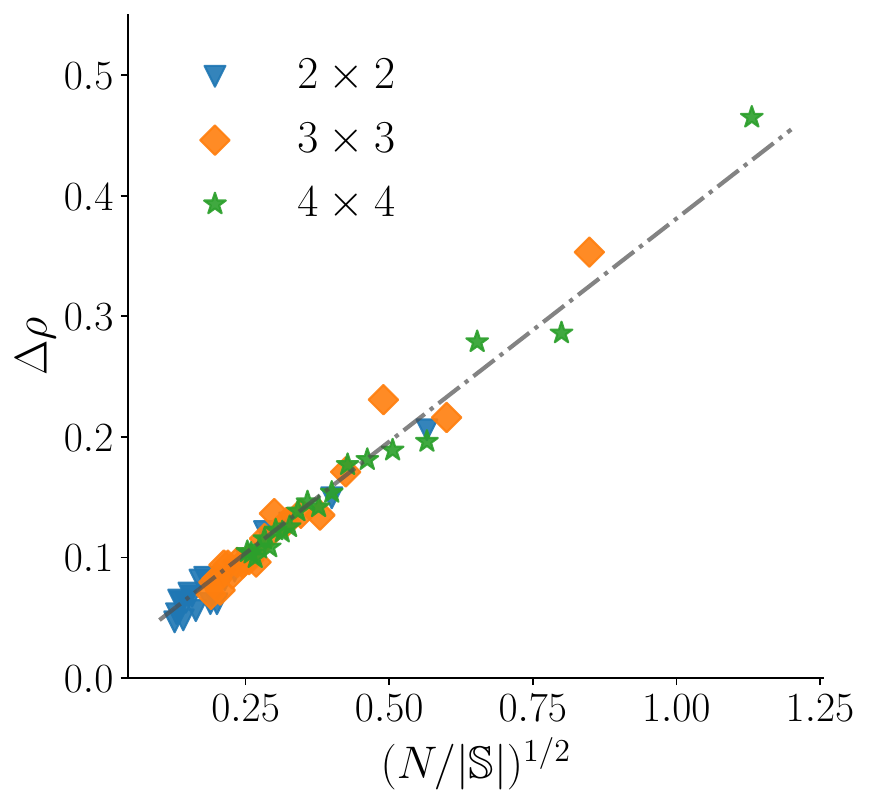}}
    \hskip 0.2cm
    \subfloat[Stone--Wales]{\label{fig:sdft:err:stone_wales}\includegraphics[height=3.8cm]{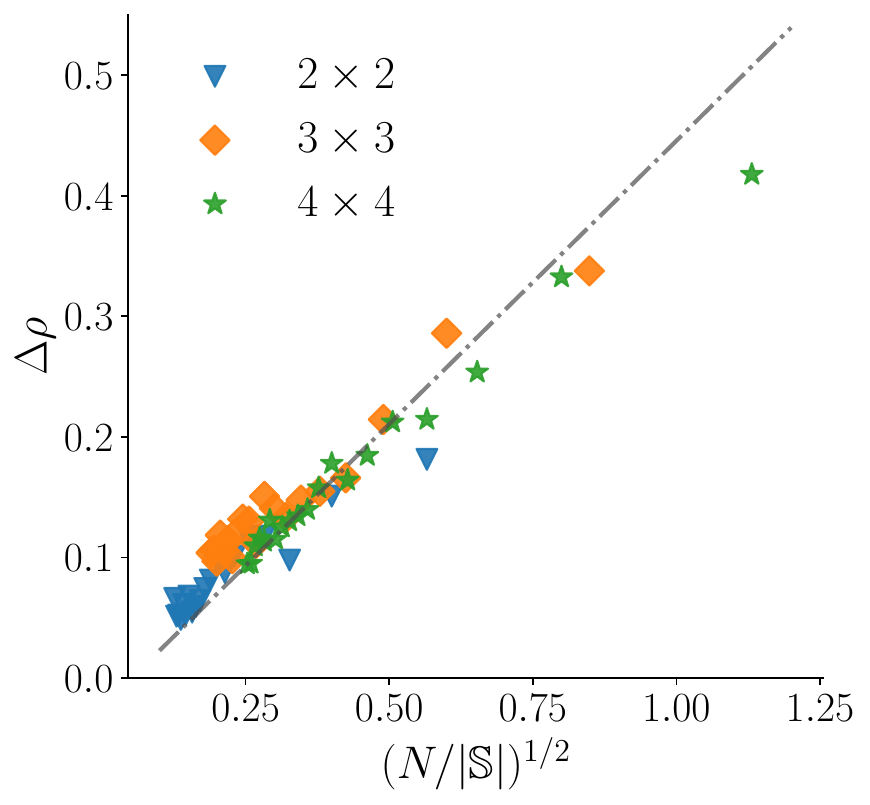}}
    \hskip 0.2cm
    \subfloat[Doping]{\label{fig:sdft:err:doping}\includegraphics[height=3.8cm]{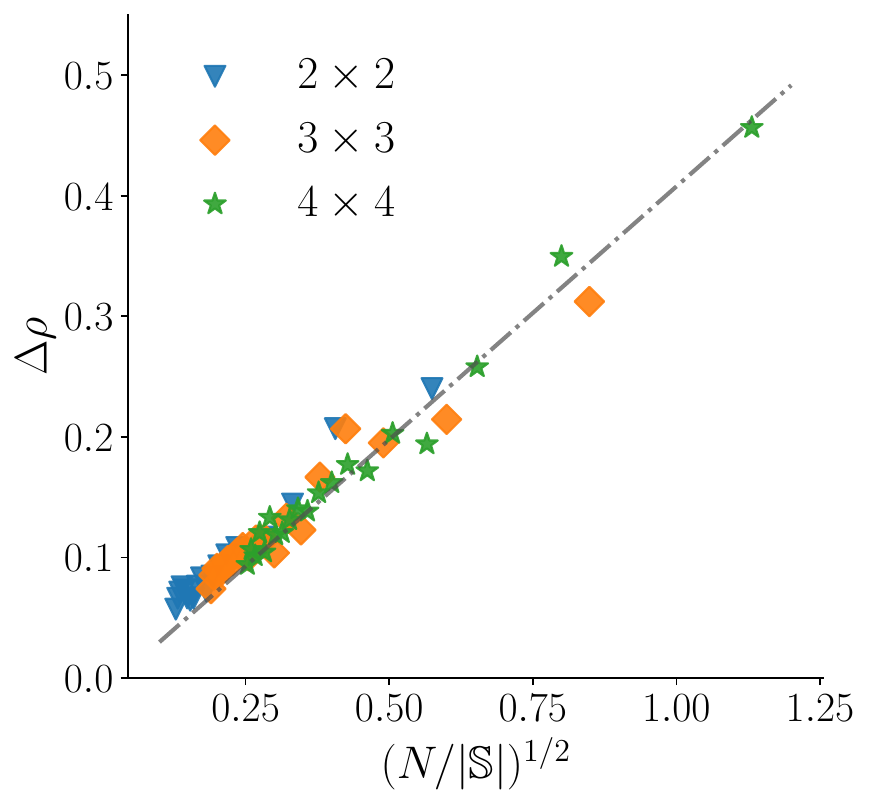}}
    \caption{The statistical error of electron density with respect to $(N/|\mS|)^{1 /2}$.} 
    \label{fig:std:rho}
\end{figure}

We finally test the performance of the MLMC methods in sDFT for the three systems, with two multilevel strategies introduced in \Cref{sec:MLMC_sDFT}.
We quantify the computational cost in terms of the number of floating-point operations determined by \cref{ml_cost}. 
{
In the numerical comparison below, the target accuracy is kept the same for all methods. 
Hence we omit the common $\epsilon$-dependent prefactor when reporting the operation counts, which is the same to all strategies and does not affect their relative cost comparison.
The level variances $\mlv_\ell$ are estimated by numerical sampling, while the level costs $\mlc_\ell$ are approximated by operation counts.
For the plane-wave cutoff hierarchy, we use $\mlc_0\approx Mn^{(0)}\log n^{(0)}$ for $\ell=0$, and $\mlc_\ell\approx M\big(n^{(\ell)} \log n^{(\ell)}+n^{(\ell-1)}\log n^{(\ell-1)}\big)$ for $\ell\geq 1$. 
For the polynomial order hierarchy, we use $\mlc_\ell\approx M^{(\ell)} n\log n$.
}

For the multilevel strategy with energy cutoffs, we fix the smearing parameter at $\beta=10^2$ ${\rm Ha}^{-1}$ and 
{consider three finest-level energy cutoffs $\Ec=20,~30,~40$ Ha. For each $\Ec$, the Chebyshev expansion is truncated at the smallest order $M$ such that $|c_{m}|\leq 10^{-3}$ for $m>M$, which gives $M\approx 420,~480,~520$, respectively.} 
We employ an algebraic growth in \cref{ecl:p} with parameters $E_0=6$ Ha, $s=0.1$, {and $p \in [0.9, 1.5]$}.
We show in \Cref{fig:mlmcec:var} the variance of both $\lmls{\ell}$ and $\lmls{\ell} - \lmls{\ell-1}$ across levels with the energy cutoff $\Ec = 40$ Ha.
It is observed that $\Var[\lmls{\ell}-\lmls{\ell-1}]$ decays exponentially as $\ell$ increases, while the variance $\Var[\lmls{\ell}]$ remains nearly unchanged across levels.
This is consistent with our discussion in \Cref{thm:mlmc_ec}, indicating that the MLMC scheme requires significantly less random orbitals than the standard sampling scheme. 
In \Cref{fig:mlmcec:cost}, we compare the computational costs {for the three finest-level} energy cutoffs $\Ec=20,~30,~40$ Ha, with {the} corresponding numbers of levels $L = 2,~3,~4$, respectively. 
We observe that the computational costs align along a line of {$NMn^{(0)}\log n^{(0)}$}, indicating that the cost is independent of $\Ec$ with the MLMC scheme.

For the multilevel strategy with polynomial orders, we fix the energy cutoff at $\Ec=20$ Ha, and {consider three smearing parameters $\beta=10^2,~10^3,~10^4$ ${\rm Ha}^{-1}$. For each choice of $\beta$, 
the Chebyshev  expansion is truncated at the smallest finest-level degree $M$ such that $|c_{m}|<5\times10^{-5}$ for $m>M$, leading to $M\approx 1300,~6000,~12200$, respectively.} 
We employ an algebraic growth in \cref{pdl:q} with parameters $M_0 = 85$, $t = 0$, {and $q\in[2,5]$. Here, $M_0$ is selected as the smallest polynomial degree satisfying $|c_m|<5\times 10^{-3}$ for $m>M_0$, corresponding to a coarse-level truncation tolerance one hundred times larger than the finest-level tolerance.}
In \Cref{fig:mlmcpd:var}, we show the variance of random matrices with the smearing parameter $\beta = 10^4~\mathrm{Ha}^{-1}$, and observe that $\Var[\lmls{\ell}-\lmls{\ell-1}]$ decays exponentially with $\ell$, while $\Var[\lmls{\ell}]$ remains nearly constant.
In \Cref{fig:mlmcpd:cost}, we compare the computational costs {for the three} smearing parameters $\beta=10^2,~10^3,~10^4$ ${\rm Ha}^{-1}$, with {the} corresponding numbers of levels $L = 2,~3,~4$, respectively.
We see that all costs align along a line {of $NM^{(0)}n\log n$}, indicating that the cost of the multilevel approach is independent of the smearing parameter $\beta$.

\begin{figure}[htbp!]
\centering
    \subfloat[Perfect lattice]{\label{fig:mlmcec:var:graphene}\includegraphics[height=3.6cm]{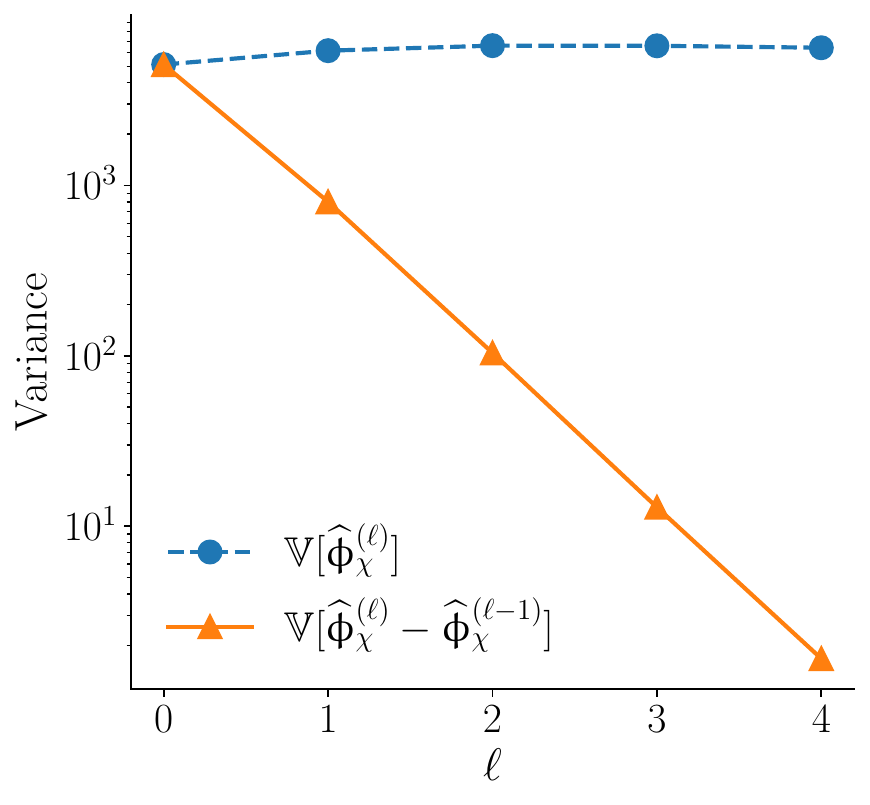}}
    \hskip 0.3cm
    \subfloat[Stone--Wales]{\label{fig:mlmcec:var:stone_wales}\includegraphics[height=3.5cm]{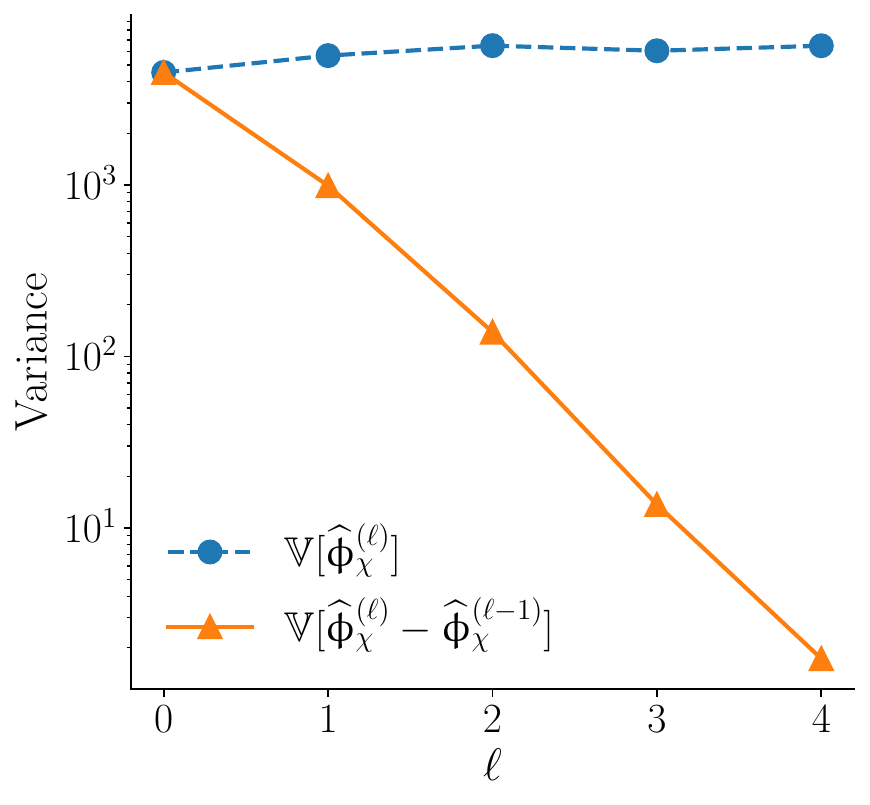}}
    \hskip 0.3cm
    \subfloat[Doping]{\label{fig:mlmcec:var:doping}\includegraphics[height=3.6cm]{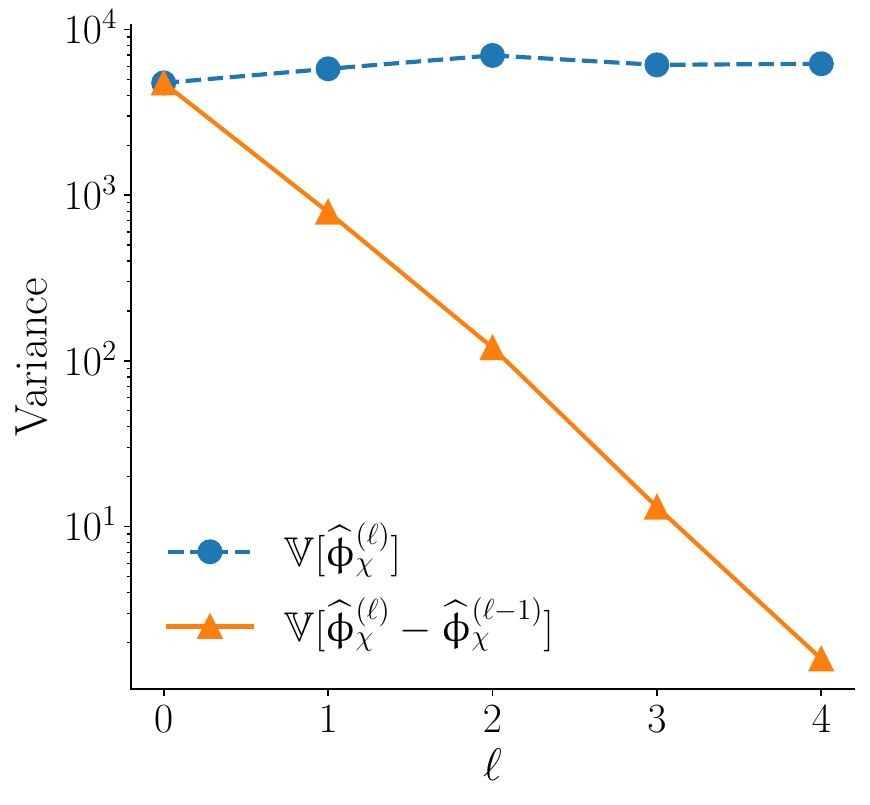}}
    \caption{Variance across levels for multilevel energy cutoffs.}
    \label{fig:mlmcec:var}
\centering
    \subfloat[Perfect lattice]{\label{fig:mlmcec:cost:graphene}\includegraphics[height=3.6cm]{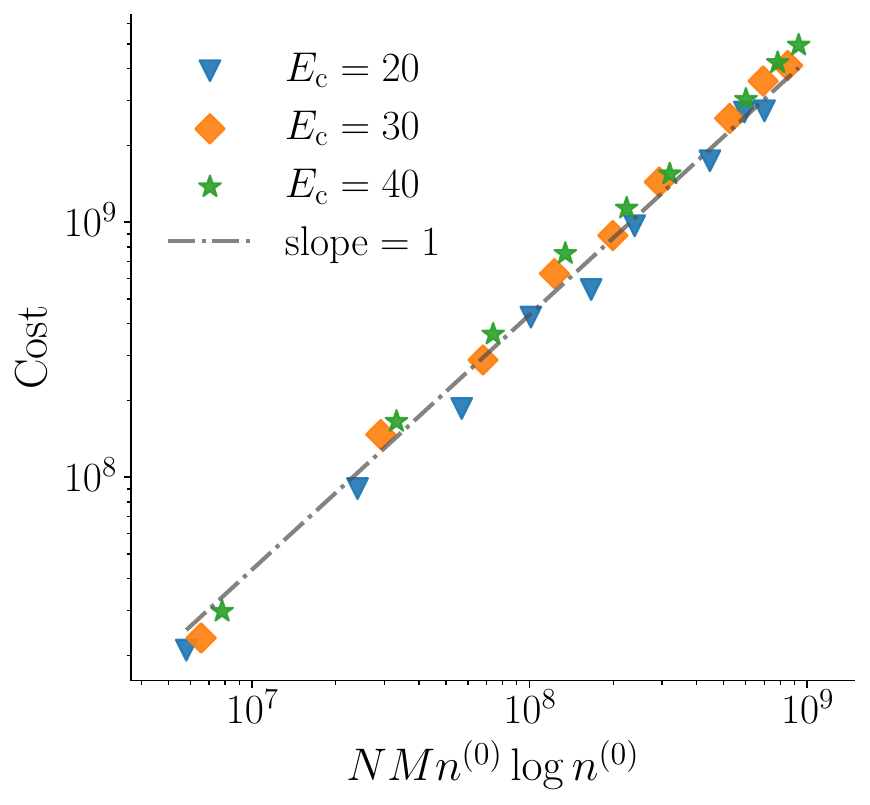}}
    \hskip 0.3cm
    \subfloat[Stone--Wales]{\label{fig:mlmcec:cost:stone_wales}\includegraphics[height=3.6cm]{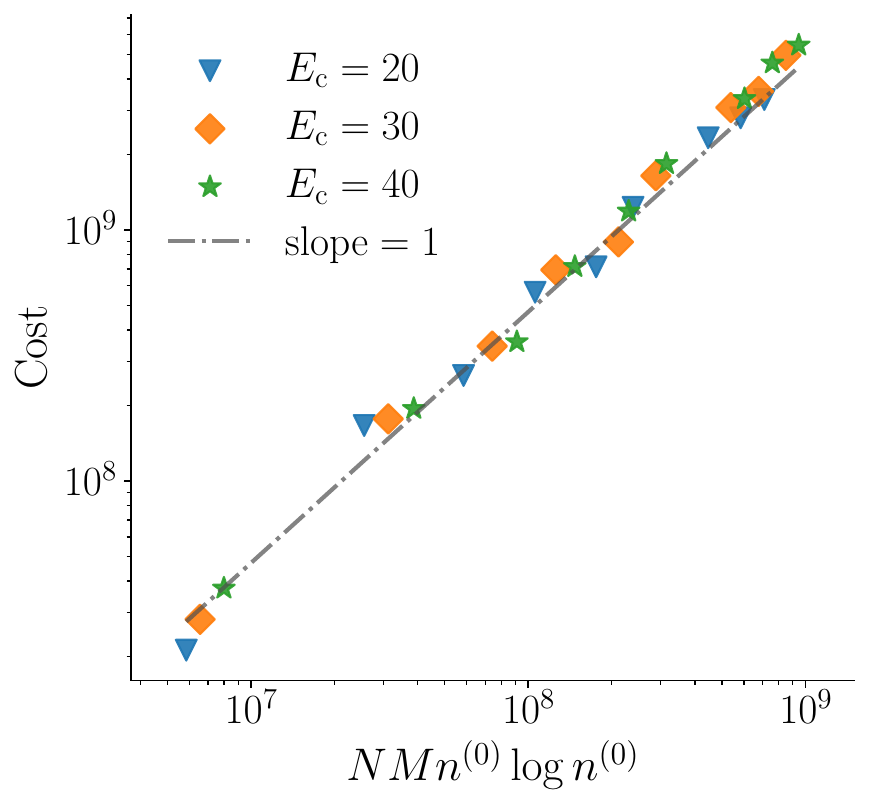}}
    \hskip 0.3cm
    \subfloat[Doping]{\label{fig:mlmcec:cost:doping}\includegraphics[height=3.5cm]{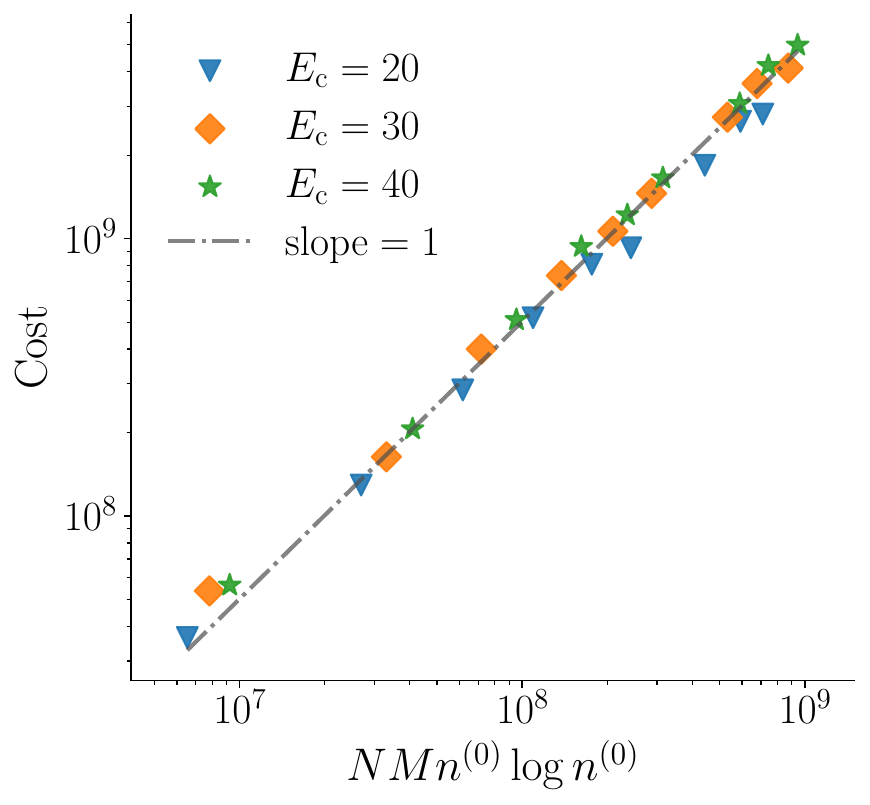}}
    \caption{Computational cost with respect to {$NMn^{(0)}\log n^{(0)}$}, for multilevel energy cutoffs.}
    \label{fig:mlmcec:cost}
\centering
    \subfloat[Perfect lattice]{\label{fig:mlmcpd:var:graphene}\includegraphics[height=3.6cm]{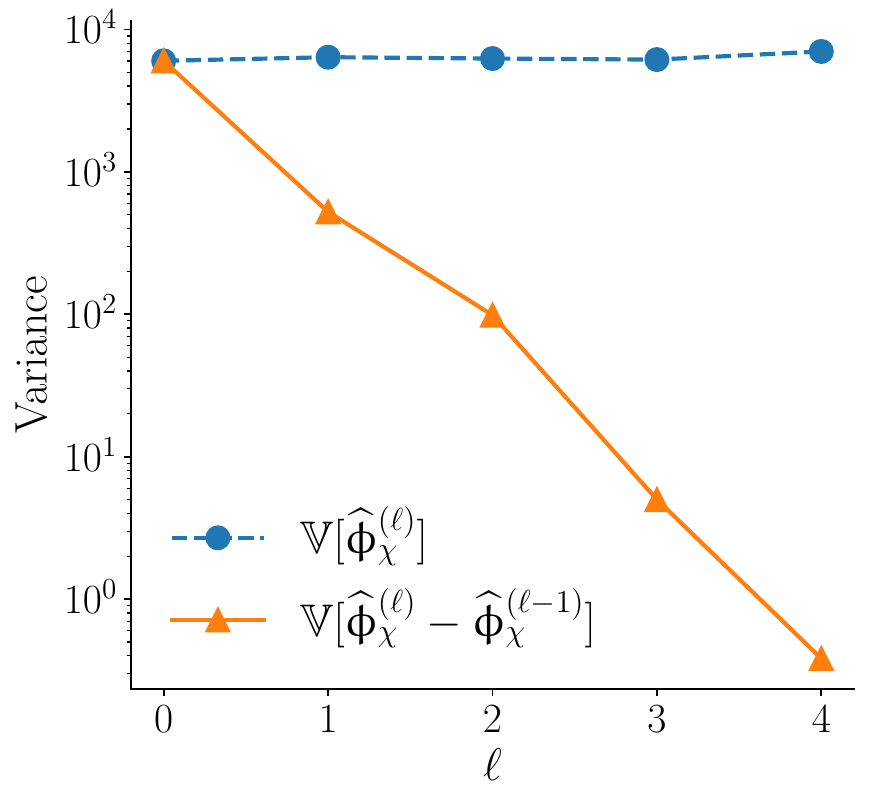}}
    \hskip 0.3cm
    \subfloat[Stone--Wales]{\label{fig:mlmcpd:var:stone_wales}\includegraphics[height=3.6cm]{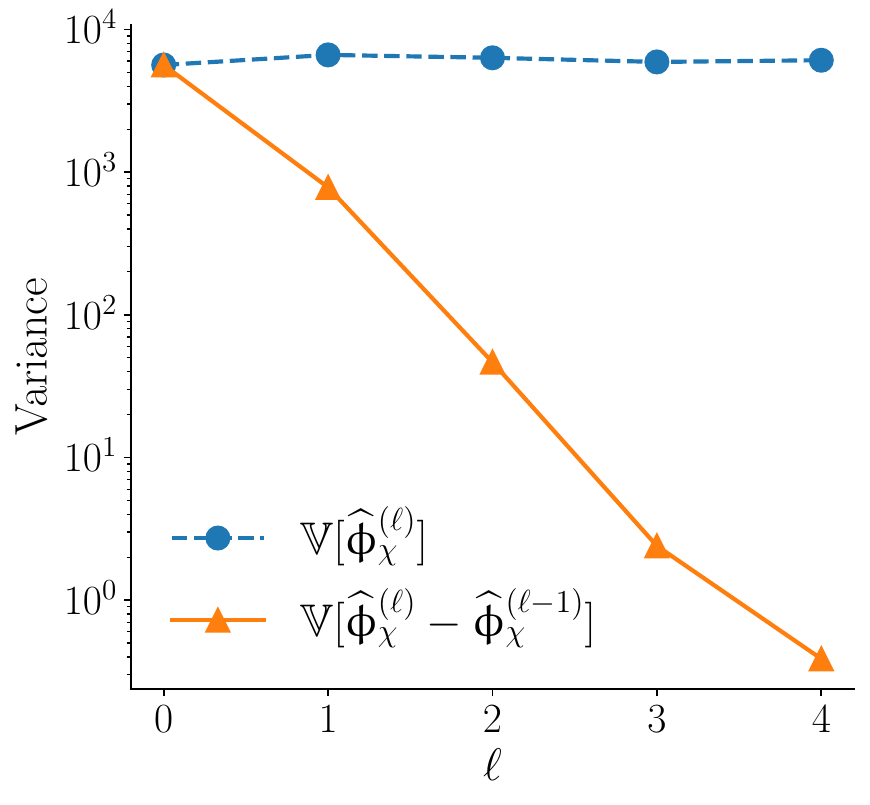}}
    \hskip 0.3cm
    \subfloat[Doping]{\label{fig:mlmcpd:var:doping}\includegraphics[height=3.6cm]{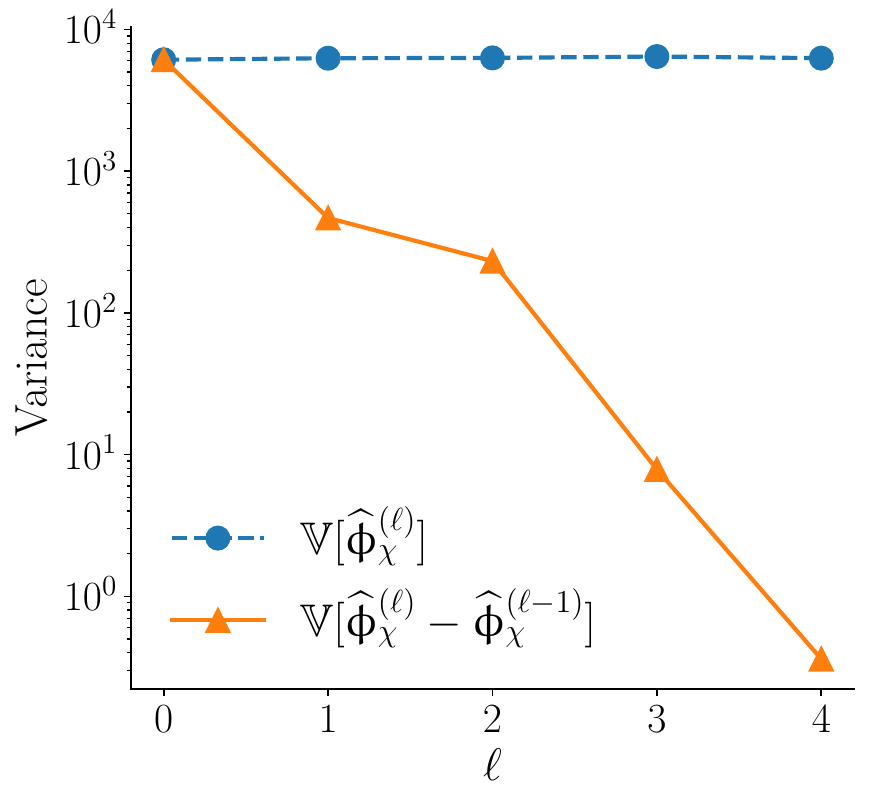}}
    \caption{Variance across levels for multilevel polynomial orders.}
    \label{fig:mlmcpd:var}
\centering
    \subfloat[Perfect lattice]{\label{fig:mlmcpd:cost:graphene}\includegraphics[height=3.6cm]{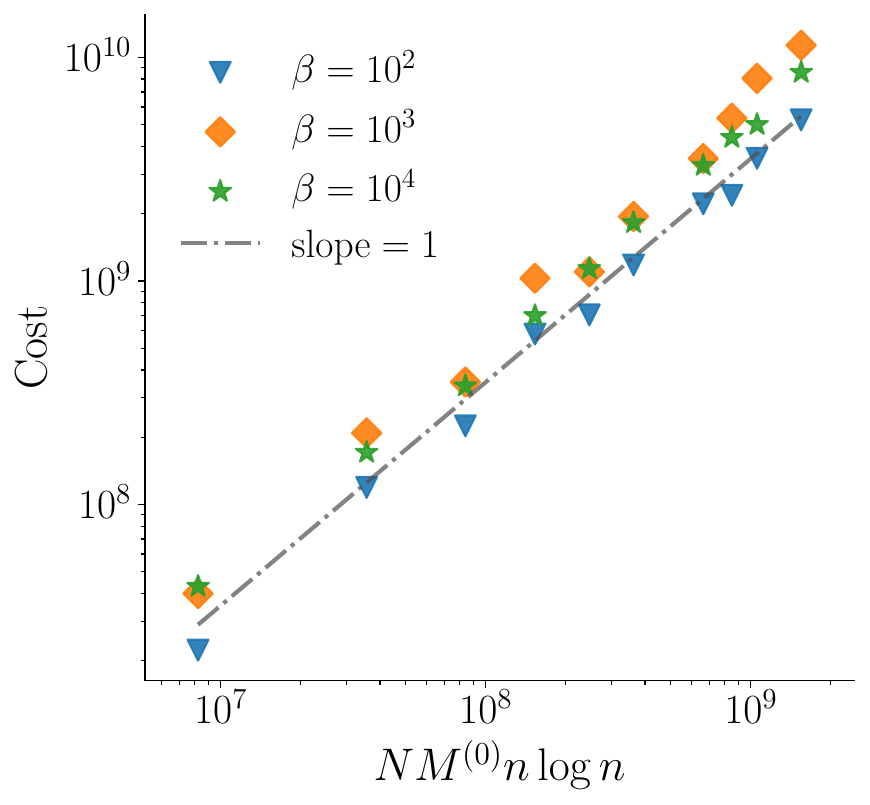}}
    \hskip 0.3cm
    \subfloat[Stone--Wales]{\label{fig:mlmcpd:cost:stone_wales}\includegraphics[height=3.6cm]{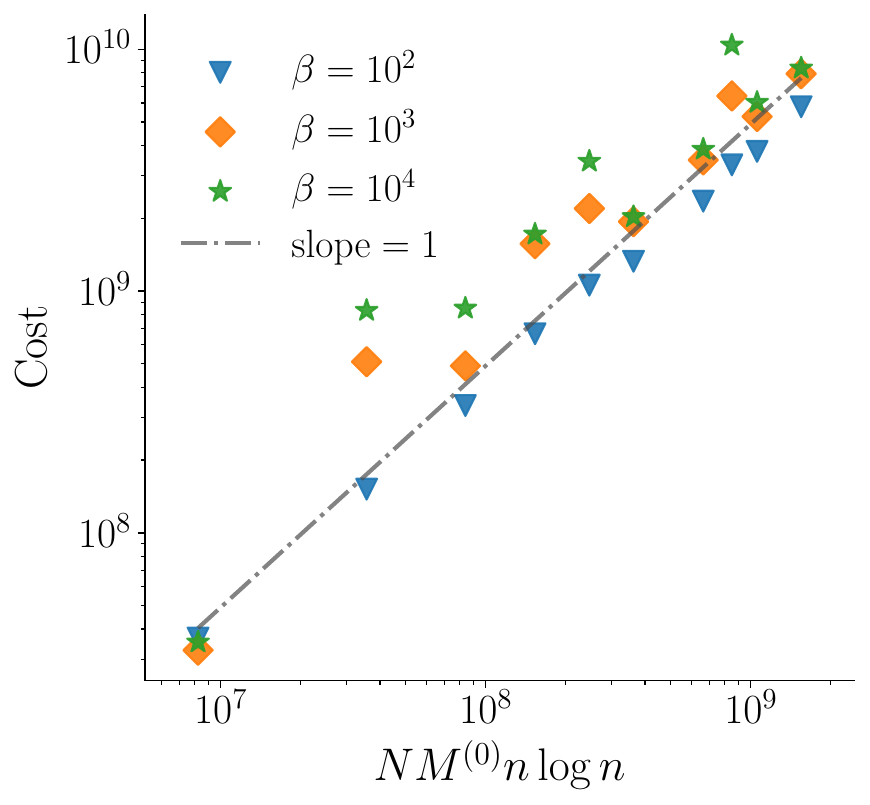}}
    \hskip 0.3cm
    \subfloat[Doping]{\label{fig:mlmcpd:cost:doping}\includegraphics[height=3.6cm]{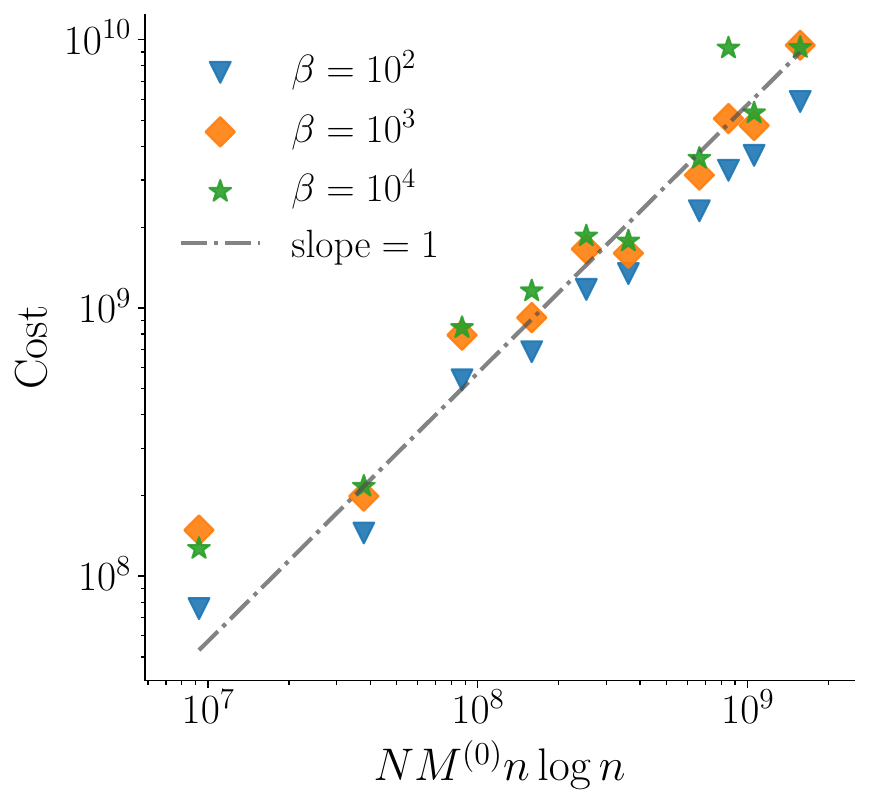}}
    \caption{Computational cost with respect to {$NM^{(0)}n\log n$}, for multilevel polynomial orders.}
    \label{fig:mlmcpd:cost}
\end{figure}

\vskip 0.2cm

\noindent
{\bf {Example 2 (3D systems).}}
{We consider 3D systems based on bulk silicon in the diamond-cubic structure: pristine silicon, boron-doped silicon (with a doping ratio of 2\%), and carbon-doped silicon (with a doping ratio of 2\%). 
The latter two systems are constructed by replacing some silicon atoms in the diamond-cubic silicon supercells with boron and carbon atoms, respectively. 
The system size $N$ is varied from 32 to 1152 through the construction of supercells with different sizes. 
All simulations for 3D systems are carried out on
one BSCC-T6 compute node equipped with two Intel Xeon Platinum 9242 CPUs at 2.30 GHz, providing 96 CPU cores and 384 GB of RAM. 
In each simulation, 24 CPU cores on a single node are used.

For DFT models, we use the analytic Goedecker--Teter--Hutter (GTH) pseudopotential together with the PBE-GGA exchange-correlation functional \cite{Perdew1996}.
In all simulations, only the $\Gamma$ point is used for $k$-point sampling in the Brillouin zone.}

{We first compare the computational efficiency of standard Kohn--Sham DFT and the two MLMC-based sDFT methods.
All tests are performed with energy cutoff $\Ec=25~{\rm Ha}$ and smearing parameter $\beta=10~{\rm Ha}^{-1}$. 
In the standard DFT calculations, the number of eigenvalues needed is chosen such that the occupation number of the highest band is below $10^{-6}$. 
For the sDFT calculations, the Chebyshev expansion order is chosen so that the truncation tolerance satisfies $|c_{m}|<10^{-6}$ for $m>M$, which gives $M=135$.
We target an $L^2$ density error of approximately $0.1$. In our tests, this is achieved by determining the number of stochastic orbitals at each level from \cref{num_orb} with sampling tolerance $\epsilon=0.5$.
For the energy-cutoff MLMC method, we use \cref{ecl:p} with $E_0$ chosen such that the corresponding number of degrees of freedom is approximately $4500$. We set $s=0.1$ and $p=1.7$, and vary the number of levels $L$ from $1$ to $4$. 
For the polynomial-order MLMC method, we use \cref{pdl:q} with $M_0$ chosen such that $|c_m|<10^{-2}$ for all $m>M_0$. We set $t=0$, $q=0.8$, and $L=2$.
The wall-time comparison and the corresponding density errors are reported in \Cref{fig:wall_time} and \Cref{fig:mlmc_error}, respectively.

\Cref{fig:wall_time} shows that the cost of standard Kohn--Sham DFT grows rapidly with system size, whereas the MLMC-based sDFT methods remain substantially more efficient for larger systems.
\Cref{fig:mlmc_error} further shows that both MLMC methods maintain the prescribed density accuracy across all three silicon systems, with only small fluctuations due to stochastic sampling. 
These results demonstrate that the proposed MLMC strategies reduce the computational cost while controlling the density error.}

{We finally compare in \Cref{fig:speed_up} the two MLMC calculations with the corresponding single-level sDFT method using the same finest discretization parameters. 
We see that both MLMC strategies provide clear speed-ups over single-level sDFT. 
The speed-up of the energy-cutoff MLMC method increases with system size, reaching about $7$--$9$ for the largest systems tested, whereas that of the polynomial-order MLMC method remains around $2$--$3$.
This difference is expected in the present setting. The energy-cutoff hierarchy reduces the spatial discretization cost on coarse levels and hence weakens the dependence on the system size $N$, whereas the polynomial-order hierarchy uses the same plane-wave basis and only reduces the polynomial degree. Since the finest polynomial order $M=135$ used here is moderate, the potential gain from the polynomial-order hierarchy is limited. We expect this hierarchy to become more beneficial in lower-temperature calculations, where a larger polynomial degree is required.
}
 
\begin{figure}[htbp!]
    \centering
    \subfloat[Pristine silicon]{\includegraphics[height=3.7cm]{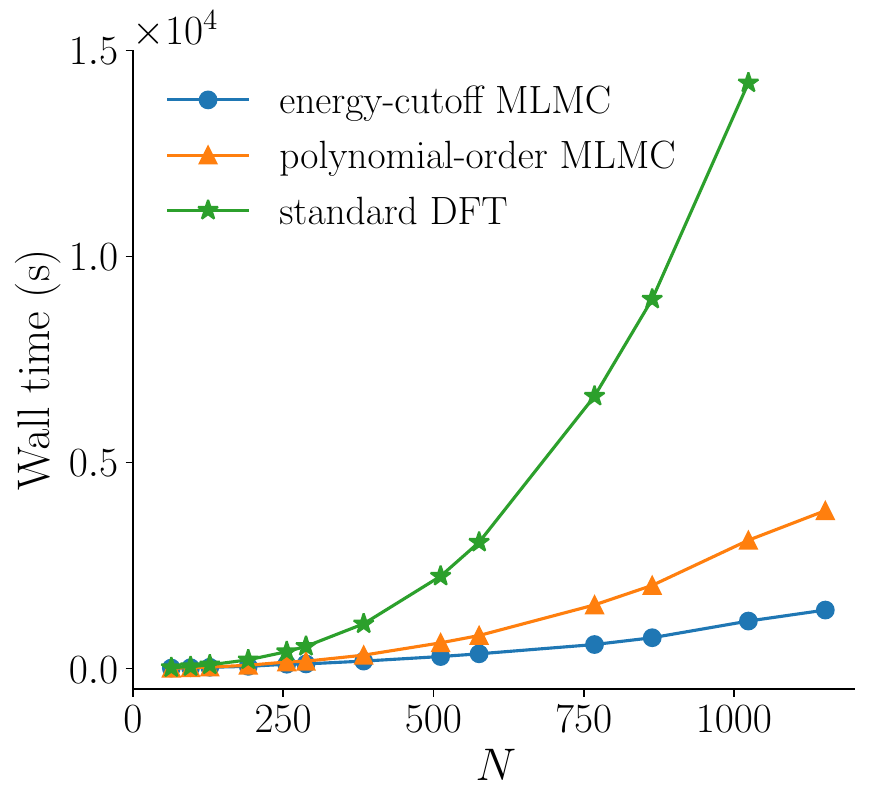}}
    \hskip 0.3cm
    \subfloat[Boron-doped silicon]{\includegraphics[height=3.7cm]{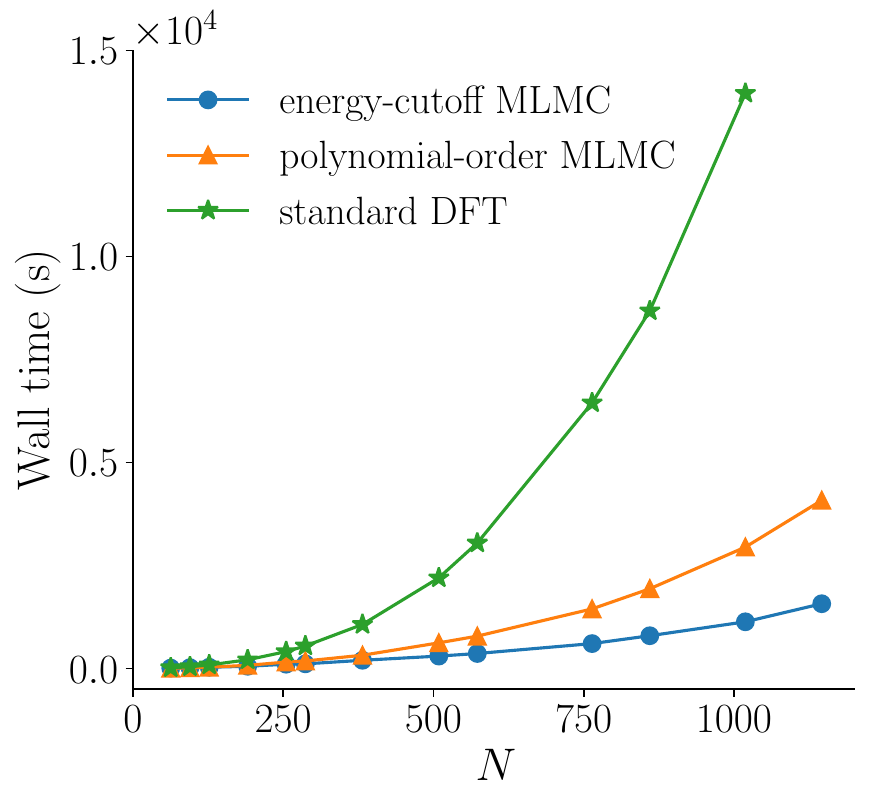}}
    \hskip 0.3cm
    \subfloat[Carbon-doped silicon]{\includegraphics[height=3.7cm]{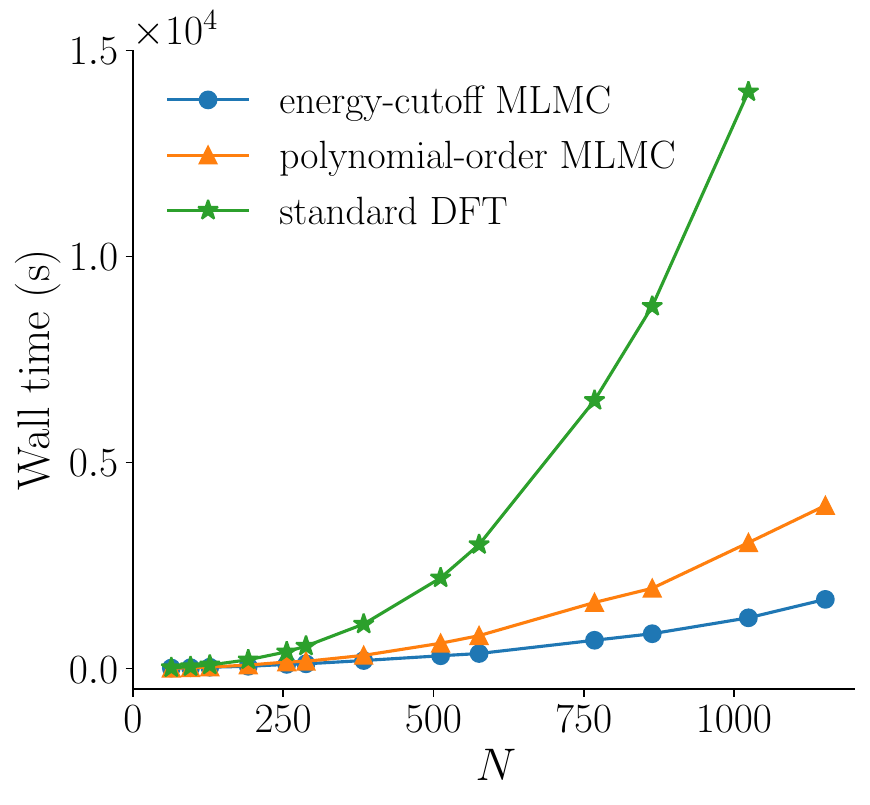}}
    \caption{Wall time comparison between standard Kohn--Sham DFT calculations and the two MLMC calculations.}
    \label{fig:wall_time}
    \centering
    \subfloat[Pristine silicon]{\includegraphics[height=3.7cm]{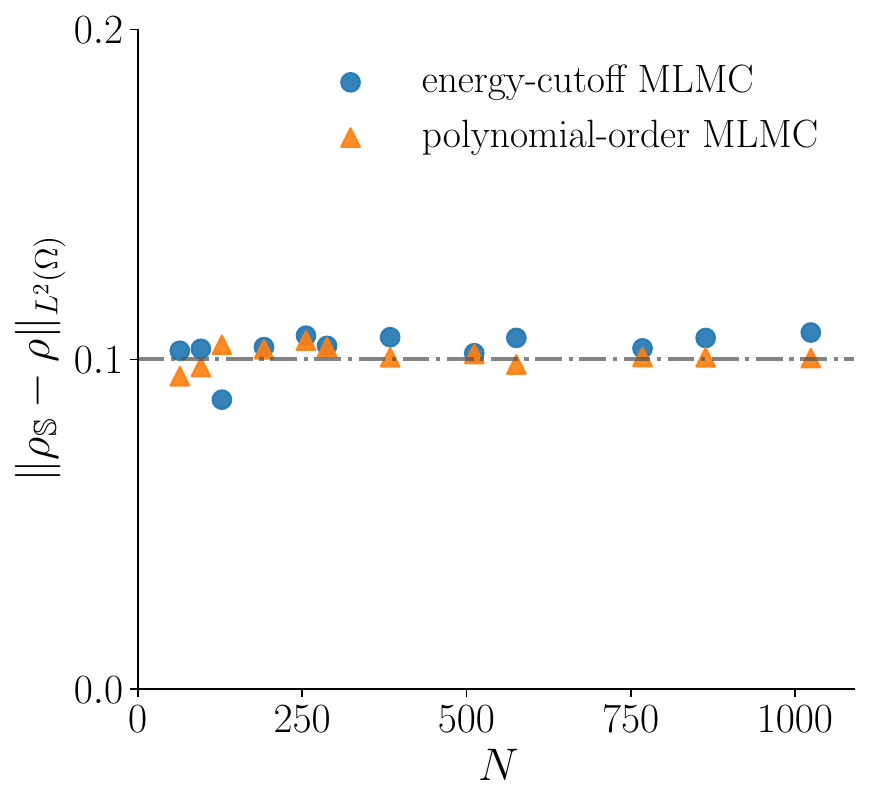}}
    \hskip 0.3cm
    \subfloat[Boron-doped silicon]{\includegraphics[height=3.7cm]{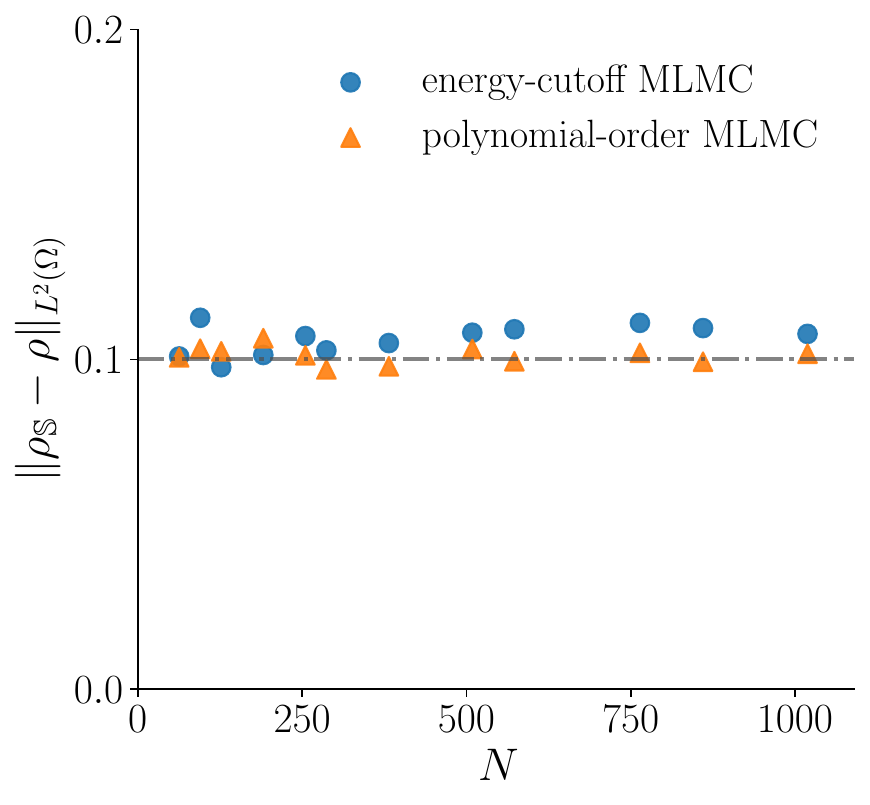}}
    \hskip 0.3cm
    \subfloat[Carbon-doped silicon]{\includegraphics[height=3.7cm]{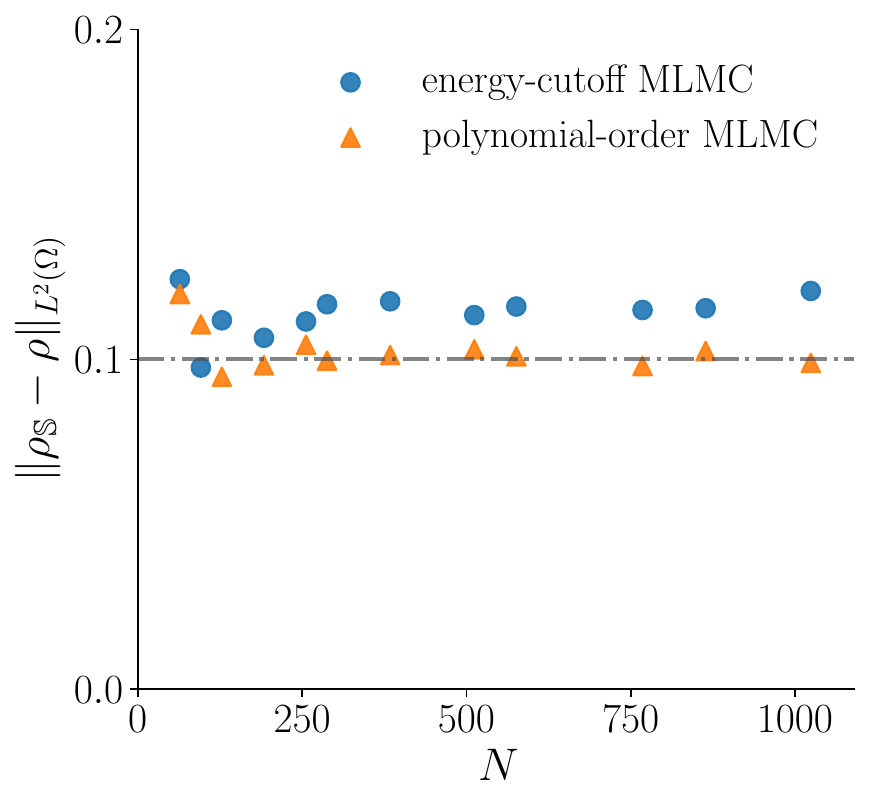}}
    \caption{The $L^2$-errors of electron density for the two MLMC calculations. 
    The results from standard Kohn--Sham DFT calculations are taken as the reference density. 
    The dashed lines indicate the target accuracy.}
    \label{fig:mlmc_error}
    \centering
    \subfloat[Pristine silicon]{\includegraphics[height=3.6cm]{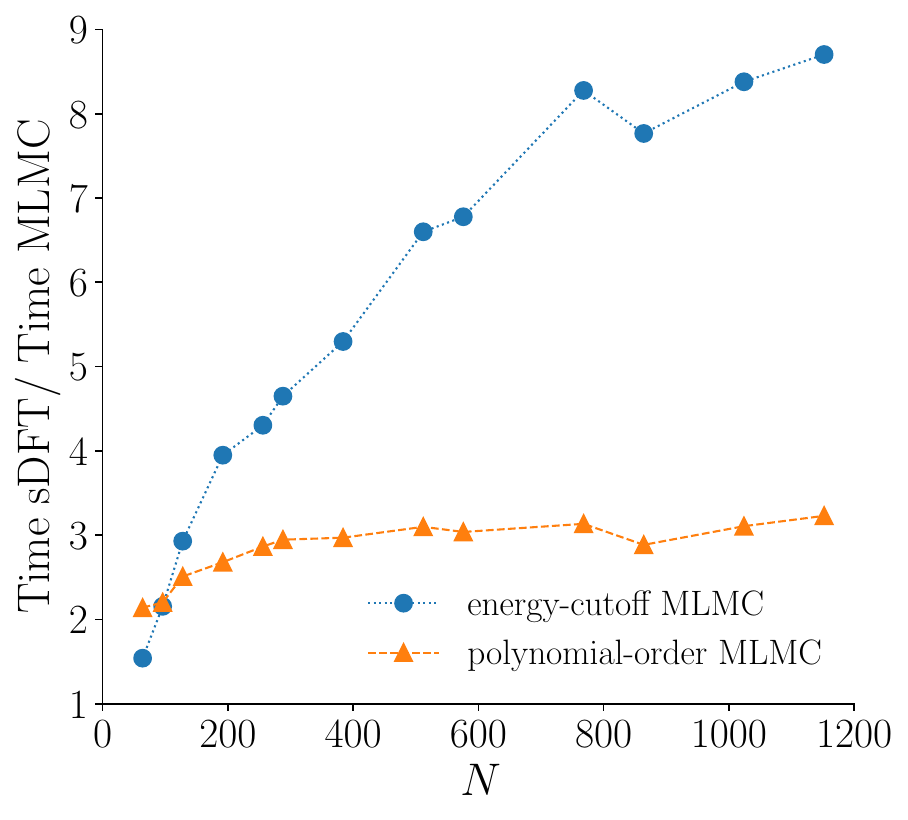}}
    \hskip 0.3cm
    \subfloat[Boron-doped silicon]{\includegraphics[height=3.6cm]{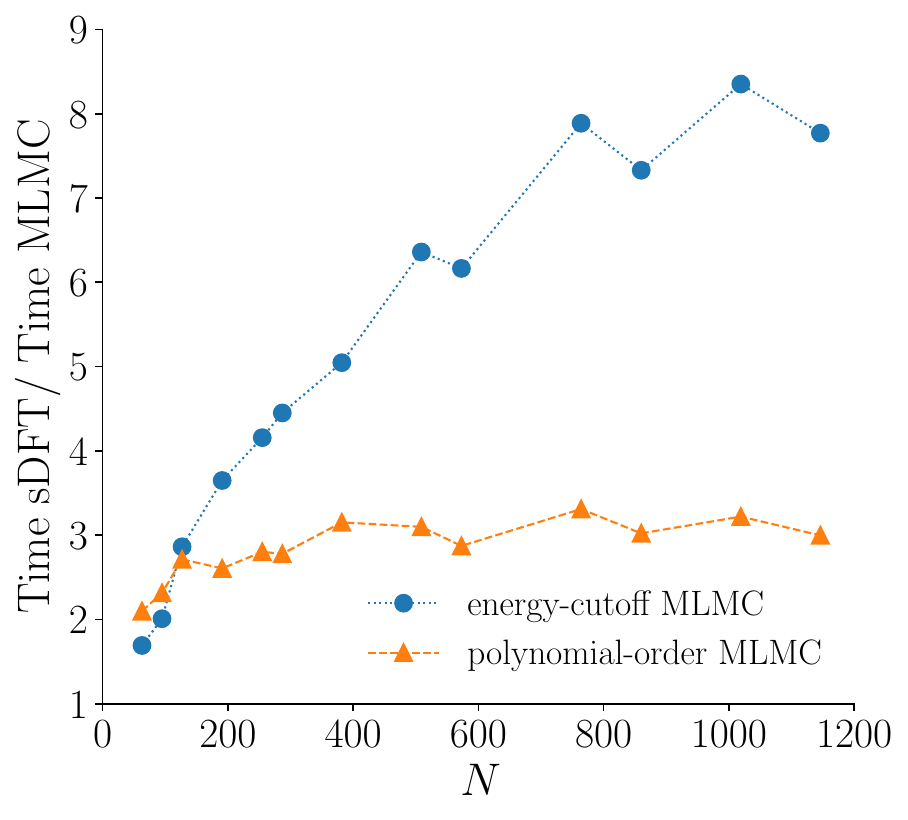}}
    \hskip 0.3cm
    \subfloat[Carbon-doped silicon]{\includegraphics[height=3.6cm]{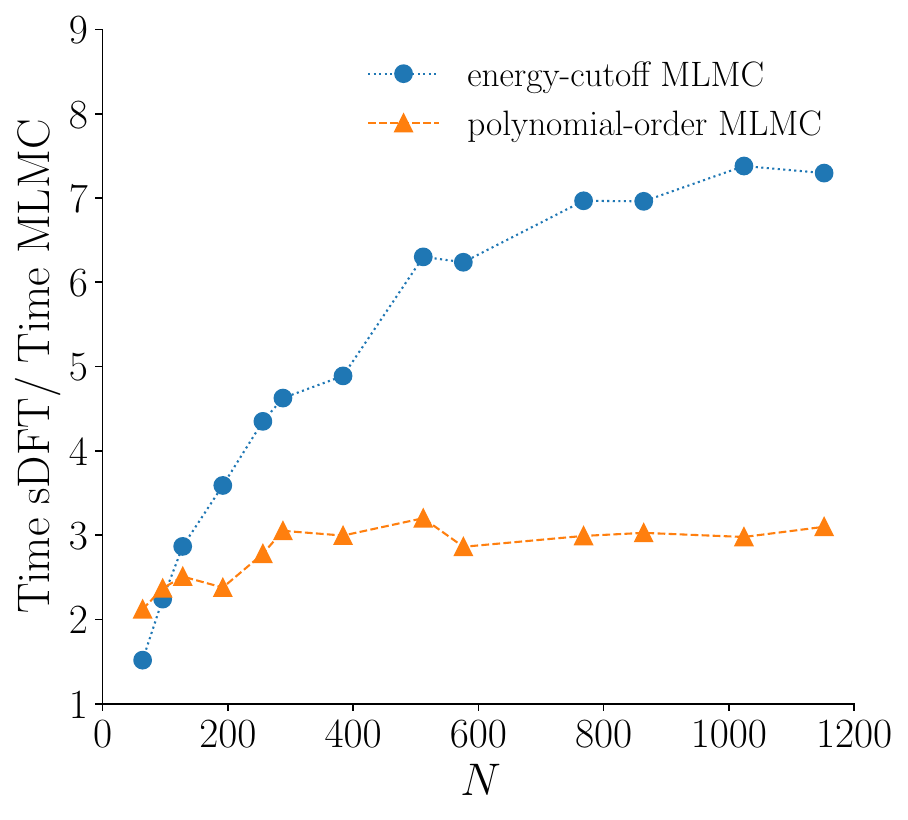}}
    \caption{Speed-up factors of the two MLMC calculations relative to the corresponding single-level sDFT method using the same finest discretization parameters. 
    Values larger than one indicate a reduction in computational time.}
    \label{fig:speed_up} 
\end{figure}

\section{Conclusions {and perspectives}}
\label{sec:conclusion}

In this work, we study the sDFT method under a plane-wave discretization and explore various variance reduction strategies within the MLMC framework. 
In particular, we develop two hierarchical approaches based on plane-wave cutoffs and Chebyshev expansion orders, which render the computational cost independent of the discretization size and the temperature, respectively. 
Our analysis also provides theoretical insights into existing variance-reduction techniques in sDFT, such as the energy-window method \cite{EWChenBaer2019} and the tempering method \cite{Nguyen2021}.

{Our future work is to combine the two hierarchies developed in this paper into a multi-index Monte Carlo framework \cite{Haji-Ali_2016}. 
In the present work, the plane-wave cutoff and the Chebyshev polynomial degree are treated separately as refinement parameters in two single-index MLMC constructions. 
A multi-index formulation would instead introduce a hierarchy indexed simultaneously by these two parameters. 
This would make it possible to balance the discretization error from the plane-wave approximation with the polynomial approximation error in the evaluation of the density matrix, and to optimize the computational budget jointly across the two directions. 
Such a formulation may provide a more efficient variance-reduction scheme than treating the two hierarchies independently, and could be further extended to other discretization parameters, such as $k$-point meshes.
}

{\section*{Acknowledgments}
This work was supported by the National Natural Science Foundation of China under grant No. 12371431.
The authors are grateful to the anonymous reviewers for their careful reading of the manuscript and their valuable suggestions, which helped improve the paper.}
% Appendices 

\appendix

\section{Proofs}
\label{sec:profs}
\renewcommand{\theequation}{A.\arabic{equation}}
\setcounter{equation}{0}

\subsection{Proof of \texorpdfstring{\Cref{lemma:var:sdft}}{lemma:var:sdft}}
\label{proof:lemma:var:sdft}

Let $\chi$ be the random orbital satisfying \cref{ass:randorb}.
Denote by $A:=\ff(H[\dm])$, $B:=AA^\ct$, and $B_\chi :=(A\chi)(A\chi)^\ct$.
Then to show \cref{var1}, it is only necessary to prove
\begin{align}
\label{proof:AB}
\Var\big[B_\chi\big]
& = \sum_{i,j=1,i\neq j}^\dof B_{ii}B_{jj}
+\sum_{i,j=1,i\neq j}^\dof B_{ij}^2\E\big[\overline{\chi}_i^2\big]\E\big[\chi_j^2\big]
+\sum_{i=1}^\dof B^2_{ii}\Big(\E\big[|\chi_i|^4\big]-1\Big) .
\end{align}

We first express the matrix elements of $B$ and $B_\chi$ explicitly
\begin{equation*}
B_{ij} =  \sum_{k}A_{ik}\overline{A}_{jk}
\quad{\rm and}\quad
(B_\chi)_{ij} = \big((A\chi)(A\chi)^\ct\big)_{ij}=\sum_{k,\ell}A_{ik}\overline{A}_{j\ell}\chi_k\overline{\chi}_\ell .
\end{equation*}
By using \cref{ass:randorb}, we can derive the expectation and variance of $(B_\chi)_{ij}$ as
\begin{align}
\nonumber
\E\big[(B_\chi)_{ij}\big] & = \sum_{k,\ell}A_{ik}\overline{A}_{j\ell}\E\big[\chi_k\overline{\chi}_\ell\big] = \sum_k A_{ik}\overline{A}_{jk} 
\qquad{\rm and}
\\[1ex]
\nonumber
\Var\big[(B_\chi)_{ij}\big] & = \E\big[\big|(B_\chi)_{ij}\big|^2\big]-\big|\E\big[(B_\chi)_{ij}\big]\big|^2
\\[1ex]
\nonumber
& = \sum_{k,\ell,s,t}\overline{A}_{ik}{A}_{j\ell}A_{is}\overline{A}_{jt}\E\big[\overline{\chi}_k\chi_\ell \chi_s\overline{\chi}_t\big]-\sum_{k,s} \overline{A}_{ik}{A}_{jk}  A_{is}\overline{A}_{js} 
\\[1ex]
\label{var:proof:1}
& = \sum_{k\neq \ell,s\neq t}\overline{A}_{ik}A_{j\ell}A_{is}\overline{A}_{jt}\E[\overline{\chi}_k\chi_\ell  \chi_s\overline{\chi}_t] + \sum_k |A_{ik}|^2|A_{jk}|^2\Big(\E\big[|\chi_k|^4\big]-1\Big).
\end{align}
To estimate the first term of \cref{var:proof:1}, we have from \cref{ass:randorb} that $\E\big[\overline{\chi}_k\chi_\ell  \chi_s\overline{\chi}_t\big] \neq 0$ only if $(k,\ell)=(s,t)$ or $(k,\ell)=(t,s)$.
This indicates
\begin{align}
\nonumber
& \sum_{k\neq \ell,s\neq t}\overline{A}_{ik}A_{j\ell}A_{is}\overline{A}_{jt}\E\big[\overline{\chi}_k\chi_\ell \chi_s\overline{\chi}_t\big]
\\[1ex]
\nonumber
= & \sum_{k\neq \ell}|A_{ik}|^2|A_{j\ell}|^2 \E\big[\overline{\chi}_k\chi_\ell \chi_k\overline{\chi}_\ell\big] + \sum_{k\neq \ell}\overline{A}_{ik}A_{j\ell}A_{i\ell}\overline{A}_{jk}\E\big[\overline{\chi}_k\chi_\ell \chi_\ell \overline{\chi}_k\big]
\\[1ex]
\label{var:proof:2}
= & \sum_{k\neq \ell}|A_{ik}|^2|A_{j\ell}|^2 + \sum_{k\neq \ell}\overline{A}_{ik}A_{j\ell}A_{i\ell}\overline{A}_{jk} \E\big[\overline{\chi}_k^2\big]\E\big[\chi_\ell^2\big] .
\end{align}

Note that $H[\dm]\in\D$ implies $A\in\D$, therefore we have $A=A^\ct$ and 
\begin{align}
\label{var:proof:3}
\sum_{i,j}|A_{ik}|^2|A_{j\ell}|^2 & = \bigg(\sum_i |A_{ki}|^2\bigg)\bigg(\sum_j |A_{\ell j}|^2\bigg)=B_{kk}B_{\ell \ell },
\\[1ex]
\label{var:proof:4}
\sum_{i,j}\overline{A}_{ik}A_{j\ell}A_{i\ell}\overline{A}_{jk} & = \bigg(\sum_i A_{ki}\overline{A}_{\ell i}\bigg)\bigg(\sum_j A_{kj}\overline{A}_{\ell j}\bigg)=B^2_{k\ell} .
\end{align}
By taking into accounts \cref{var:proof:1}-\cref{var:proof:4}, we obtain
\begin{align*}
\nonumber
\Var\big[B_\chi\big] = \sum_{i,j}\Var\big[(B_\chi)_{ij}\big] 
= \sum_{k\neq \ell}B_{kk}B_{\ell\ell} + \sum_{k\neq \ell}B^2_{k\ell}\E\big[\overline{\chi}_k^2\big]\E\big[\chi_\ell ^2\big] +\sum_k B^2_{kk} \Big(\E\big[|\chi_k|^4\big]-1\Big) 
\end{align*}
which completes the proof of \eqref{proof:AB} (and hence \cref{var1}).

The result \cref{var2} follows immediately from \cref{var1} if the random orbital $\chi_i$ additionally satisfy the conditions in \cref{chi:ass:2}.
This completes the proof.
\hfill$\proofbox$

\subsection{Proof of \texorpdfstring{\Cref{lemma:ml_var}}{lemma:ml:var}}
\label{proof:lemma:ml_var}

For simplicity of the presentation, we will denote $\PE(H[\dm])$ by $\PM$ in this proof.
We have from \cref{mlvell} that 
\begin{align}
\label{var_ml2}
\mlv_\ell = \Var\big[\lmls{\ell}-\lmls{\ell-1}\big]
\leq \left(\sqrt{\Var\big[\lmls{\ell}-\pmp\big(\dm,\chi,\PE\big)\big]} + \sqrt{\Var\big[\lmls{\ell-1}-\pmp\big(\dm,\chi,\PE\big)\big]}\right)^2.
\end{align}
By using the formulas \cref{dm:rand:2} for $\pmp\big(\dm,\chi,\PE\big)$ and \cref{sdm_l} for $\lmls{\ell}$, we have 
\begin{align*}
& \Var\big[\lmls{\ell}-\pmp\big(\dm,\chi,\PE\big)\big] 
\\[1ex]
= \;& \Var\Big[ \big(\ql{\ell}\chi\big)\big(\ql{\ell}\chi\big)^\ct - \big(\PM\chi\big)\big(\PM\chi\big)^\ct \Big] 
\\[1ex]
=\;&\Var\Big[\Big(\ql{\ell}\chi\Big)\Big(\big(\ql{\ell}-\PM\big)\chi\Big)^\ct+\Big(\big(\ql{\ell}-\PM\big)\chi\Big)\Big(\PM\chi\Big)^\ct\Big] 
\\[1ex]
\leq\; & \left(\sqrt{\Var\Big[\Big(\ql{\ell}\chi\Big)\Big(\big(\ql{\ell}-\PM\big)\chi\Big)^\ct\Big]} + \sqrt{\Var\Big[\Big(\big(\ql{\ell}-\PM\big)\chi\Big)\Big(\PM\chi\Big)^\ct\Big]}\right)^2 .
\end{align*}
{For $\chi\in\C^n$ additionally satisfying \cref{chi:ass:2},} by using a proof similar to that of \Cref{lemma:var:sdft}, we can obtain
\begin{align*}
\Var\Big[\Big(\ql{\ell}\chi\Big)\Big(\big(\ql{\ell}-\PM\big)\chi\Big)^\ct\Big]
& \leq \tr\Big(\big(\ql{\ell}\big)^2\Big)\tr\Big(\big(\ql{\ell}-\PM\big)^2\Big)
\qquad{\rm and}
\\[1ex]
\Var\Big[ \Big(\big(\ql{\ell}-\PM\big)\chi\Big)\Big(\PM\chi\Big)^\ct\Big]
& \leq \tr\big(\Phi_M(\dm)\big)\tr\Big(\big(\ql{\ell}-\PM\big)^2\Big).
\end{align*}
Since $\tr\Big(\big(\ql{\ell}\big)^2\Big)=\Oc(N)$ and $\tr\big(\Phi_M(\dm)\big)=\Oc(N)$, we have
\begin{align*}
\Var[\lmls{\ell}-\pmp\big(\dm,\chi,\PE\big)] 
\leq C N \big\|\ql{\ell}-\PM\big\|^2_{\F}.
\end{align*}
This together with \cref{var_ml2} implies
\begin{align*}
\mlv_\ell \leq C N\left(\big\|\ql{\ell}-\PM\big\|_{\F} + \big\|\ql{\ell-1}-\PM\big\|_{\F}\right)^2
\leq C N\big\|\ql{\ell-1}-\PM\big\|^2_{\F},
\end{align*}
where the second inequality follows from the fact that the accuracy of $\ql{\ell}$ increases with $\ell$. 
This completes the proof.
\hfill$\proofbox$

\subsection{Proof of \texorpdfstring{\Cref{thm:mlmc_ec}}{thm:mlmc:ec}}
\label{proof:thm:mlmc_ec}

According to \Cref{lemma:ml_var}, the variance at level $\ell$ is controlled by the approximation error of $\ql{\ell-1}$. 
{Therefore, the proof of \Cref{thm:mlmc_ec} amounts to deriving an exponentially decaying error estimate for $\ff(\Hlm[\dm])$ as the energy cutoff $\Ecl{\ell-1}$ increases. This relies on two observations of its matrix elements: convergence with respect to the cutoff and exponential off-diagonal decay. These are established in \Cref{app:lemma:trunc_converge} and \Cref{app:lemma:G_decay}, respectively.

In the rest of this section, we will denote $H[\dm]$ by $H$ for simplicity of the presentation. 
Let $\widetilde{\Hlm}$ be the zero-padded matrix of $\Hlm$ over $\C^{n\times n}$, that is
\begin{equation} 
\label{Htrun}
\widetilde{\Hlm}_{\G\G'} := 
\left\{\begin{array}{ll} H_{\G\G'} \qquad &{\rm if}~ |\G|^2,|\G'|^2\leq 2\Ecl{\ell-1},
\\[1ex]
0 & {\rm otherwise}.
\end{array}
\right. 
\end{equation}
}

\begin{lemma}
\label{app:lemma:trunc_converge}
There exist positive constants $C$ and $\gamma$ depending on $\beta$ such that 
\begin{equation}
\label{cut_decay}
\left|{\ff(\widetilde{\Hlm})}_{\G\G'} - \ff(H)_{\G\G'}\right|\leq C e^{-\gamma \sqrt{\Ecl{\ell-1}}}
\end{equation}
for any $1\leq\ell\leq L$ and $|\G|^2, |\G'|^2\leq \Ecl{\ell-1}$.
\end{lemma}

\begin{proof}
Let $\mathfrak{d}(\cdot)$ denote the spectrum of a matrix.
Then we have from \eqref{Htrun} that $\mathfrak{d}(\widetilde{\Hlm}) = \mathfrak{d}(\Hlm)\cup\{0\}$.
Let $\cont$ be a contour that encloses the spectrum of $\widetilde{\Hlm}$ and $H$, and avoids all the singularities of $\ff$, while satisfying 
\begin{equation}
\label{cont_dist}
\min\left\{{\rm dist}\big(\cont,\mathfrak{d}(H)\big),~ {\rm dist}\big(\cont,\mathfrak{d}(\widetilde{\Hlm})\big)\right\}\geq\sigma_\beta .
\end{equation}
Here $\sigma_\beta$ depends on $\beta$.

Using the {exponential decay assumption for the Fourier coefficients} of the potentials of the Kohn--Sham Hamiltonian in \cref{KSham}, we have that there exists some $\gamma_h>0$ such that
\begin{equation}
\label{decay:H}
\Big|\widetilde{\Hlm}_{\G\G'}\Big|\leq Ce^{-\gamma_h|\G-\G'|}\quad{\rm and }\quad\Big|H_{\G\G'}\Big|\leq Ce^{-\gamma_h|\G-\G'|} .
\end{equation}
Then by using \cref{cont_dist}, \cref{decay:H} and a Combes-Thomas type estimate \cite{combes1973asymptotic} (see also similar arguments in \cite[Lemma 6]{ChenOrtner2017}), we have that there exists constants $C>0$ and $\gamma_{h,\beta}>0$ depending on $\beta$, such that for any $z\in\cont$,
\begin{equation}
\label{RH}
\Big|\big(z-\widetilde{\Hlm}\big)^{-1}_{\G\G'}\Big| \leq C e^{-\gamma_{h,\beta}|\G-\G'|} , 
~~
\Big|\big(z-H\big)^{-1}_{\G\G'}\Big| \leq Ce^{-\gamma_{h,\beta}|\G-\G'|} .
\end{equation}

Note that \cref{Htrun} and \cref{decay:H} implies
\begin{equation*} 
\Big|\big(H-\widetilde{\Hlm}\big)_{\G''\G'''}\Big| \leq
\left\{\begin{array}{ll}
0\qquad &{\rm if}~ |\G''|^2,|\G'''|^2\leq 2 \Ecl{\ell-1},
\\[1ex]
Ce^{-\gamma_h|\G''-\G'''|} & {\rm otherwise}.
\end{array}
\right. 
\end{equation*}
This together with \cref{RH} and $|\G|^2, |\G'|^2\leq 2\Ecl{\ell-1}$ implies that
\begin{align*}
& \Big|\big(z-\widetilde{\Hlm}\big)^{-1}_{\G\G'}-\big(z-H\big)^{-1}_{\G\G'}\Big|
\\[1ex]
\leq\, &\sum_{\substack{|\G''|^2\leq 2\Ec\\|\G'''|^2\leq 2\Ec}}\Big|\big(z-\widetilde{\Hlm}\big)^{-1}_{\G\G''} \big(H-\widetilde{\Hlm}\big)_{\G''\G'''} \big(z-H\big)^{-1}_{\G'''\G'}\Big|
\\[1ex]
\leq\,& C\Bigg(\sum_{\substack{2\Ecl{\ell-1}<|\G''|^2\leq 2\Ec\\|\G'''|^2\leq 2\Ec}}
+ \sum_{\substack{|\G''|^2\leq 2 \Ecl{\ell-1}\\2\Ecl{\ell-1}<|\G'''|^2 \leq 2\Ec}}\Bigg) e^{-\big(\gamma_{h,\beta}|\G-\G''| + \gamma_h|\G''-\G'''| + \gamma_{h,\beta}|\G'''-\G'|\big)}
\\[1ex] 
\leq\,& C \sum_{2\Ecl{\ell-1}<|\G''|^2\leq 2\Ec} e^{-\gamma\big(|\G-\G''| + |\G'-\G''|\big)} 
~ \leq ~ C e^{-\gamma \sqrt{\Ecl{\ell-1}}} ,
\end{align*}
where the constants $C$ and $\gamma$ depend on $\beta$.
This together with the contour integral representation implies
\begin{align*}
& \left|{\ff(\widetilde{\Hlm})}_{\G\G'} - \ff(H)_{\G\G'}\right|
\\[1ex]
\leq & C\oint_{\cont}|\ff(z)|\cdot\bigg|\big(z-\widetilde{\Hlm}\big)^{-1}_{\G\G'} - \big(z-H\big)^{-1}_{\G\G'}\bigg|\dd z
\leq C e^{-\gamma \sqrt{\Ecl{\ell-1}}}\oint_{\cont}|\ff(z)|\dd z .
\end{align*}

Since the Hamiltonian $H$ is bounded from below and $\ff$ exhibits an exponential decay, we have that $\oint_{\cont}|\ff(z)|\dd z$ is bounded by some constant independent of $\Ecl{\ell}$. 
Therefore, the estimate \cref{cut_decay} holds and the proof is completed. 
\end{proof}

\vskip 0.2cm

\begin{lemma}
\label{app:lemma:G_decay}
There exist positive constants $C$ and $\gamma$ depending on $\beta$, such that 
\begin{equation}
\label{G_decay}
\left|\ff(\Hlm)_{\G\G'}\right|\leq C e^{-\gamma \left(\min\{|\G|,|\G'|\}+|\G-\G'|\right)}
\end{equation}
for any $1\leq\ell\leq L + 1$.
\end{lemma}

\begin{proof}
Without loss of generality, we assume that $|\G|\leq|\G'|$. 
Let $\bR := |\G|/2$, we will prove \cref{G_decay} by considering two cases $|\G-\G'|<\bR/2$ and $|\G-\G'|\geq\bR/2$, respectively.
    
For the first case with $|\G-\G'|<\bR/2$, we have $|\G'|\geq|\G|-\bR/2>\bR$.
Define
\begin{equation*}
\Hlm_{\bR} :=  \left\{\begin{array}{ll}
H_{\G\G'}\qquad &{\rm if}~ \bR^2/4<|\G|^2,|\G'|^2\leq 2\Ecl{\ell-1},
\\[1ex]
0 & {\rm otherwise}.
\end{array}
\right. 
\end{equation*}
We then observe that
\begin{equation}
\label{case1:0}
\Big|\ff(\Hlm)_{\G\G'}\Big| \leq \Big|\ff(\Hlm)_{\G\G'} - \ff(\Hlm_\bR)_{\G\G'}\Big| + \Big|\ff(\Hlm_\bR)_{\G\G'}\Big|.
\end{equation}
To estimate the first term of \eqref{case1:0}, we can employ a similar proof to that of \Cref{app:lemma:trunc_converge} and obtain
\begin{align*}
\Big|\ff(\Hlm)_{\G\G'} - \ff(\Hlm_\bR)_{\G\G'}\Big| \leq C e^{-\gamma\bR}.
\end{align*}
This together with $|\G-\G'|<\bR/2= |\G|/4$ leads to 
\begin{equation}
\label{case1:1}
\Big|\ff(\Hlm)_{\G\G'} - \ff(\Hlm_\bR)_{\G\G'}\Big| \leq Ce^{-\gamma(|\G|+|\G-\G'|)}.
\end{equation}
To estimate the second term of \eqref{case1:0}, we recall that the Hamiltonian is scaled and shifted by constants $a$ and $b$ respectively (introduced at the beginning of \Cref{sec:model}).
Then we have from the decay estimate \eqref{decay:H} that there exists a constant $C_V>0$ such that $\sum_{\G'\neq\G''} \big|(\Hlm_\bR)_{\G'\G''}\big| \leq |a|C_V$. 
By using the Gershgorin circle theorem, we have 
\begin{equation*}
\mathfrak{d}(\Hlm_\bR)\subset\bigcup_{\bR^2/4<|\G''|^2\leq 2\Ecl{\ell-1}}\left\{\lambda\in\R:\left|\lambda-a\Big(\frac{1}{2}|\G''|^2-b\Big)\right|<|a|C_V\right\}.
\end{equation*}
Therefore, there exists a contour $\cont$ enclosing the spectrum of $\Hlm_\bR$ that avoids all the singularities of $\ff$, and satisfies ${\rm dist}\big(\cont,\mathfrak{d}\big(\Hlm_\bR\big)\big) > \sigma_\beta$ and for any $z\in\cont$
\begin{align*}
\bigg|\frac{1}{a}\re(z) + b\bigg| \geq \min_{\bR^2/4<|\G''|^2 \leq 2\Ecl{\ell-1}}\frac{1}{2}|\G''|^2-C_V-1 > \frac{1}{8}\bR^2-C_V-1> \frac{1}{8}\bR-\bar{C}_V,
\end{align*}
where the constant $\bar{C}_V$ only depends on $C_V$. 
Note that $\ff$ decays exponentially, i.e., $|\ff(z)| \leq C e^{-\zeta_\beta |\frac{z}{a} +b|}$ for some $\zeta_\beta$ depending on $\beta$. 
This together with the estimate \cref{RH} implies
\begin{align}
\label{case1:2}
\left|\ff\big(\Hlm_\bR\big)_{\G\G'}\right|
& \leq C\oint_\cont|\ff(z)| \cdot \left| \big(z-\Hlm_\bR\big)^{-1}_{\G\G'} \right| \dd z
\\[1ex]
\nonumber
&\leq Ce^{-\zeta_\beta \big(\frac{\bR}{8} - \bar{C}_V\big)} e^{-\gamma\big(|\G-\G'|\big)}
~\leq~ Ce^{-\gamma\big(|\G|+|\G-\G'|\big)} .
\end{align}
Combining \cref{case1:0}, \eqref{case1:1}, \cref{case1:2}, we obtain the estimate \cref{G_decay}.
     
We then consider the second case with $|\G-\G'|\geq\bR/2$.
By using \cref{RH} and the contour integral representation, we have 
\begin{equation*}
\left|\ff\big(\Hlm\big)_{\G\G'}\right|\leq Ce^{-{\gamma}|\G-\G'|},
\end{equation*}
which together with $|\G-\G'|\geq\bR/2=|\G|/4$ leads to \cref{G_decay}.
\end{proof}

\vskip 0.2cm

We are now ready to prove \Cref{thm:mlmc_ec}.
\vskip 0.2cm
Let $\ql{\ell-1}$ be given by \cref{ec_h}. 
Using \Cref{lemma:ml_var}, it is only necessary for us to prove 
\begin{equation}
\label{proof-A3-1}
\left\|\ql{\ell-1}-\PE\big(H[\dm]\big)\right\|_{\F}^2 \leq C \Big(e^{-2\gamma\sqrt{\Ecl{\ell-1}}} + e^{-2\alpha M} \Big)
\qquad \forall~1\leq\ell\leq L .
\end{equation}
Then we can decompose the left-hand side of \eqref{proof-A3-1} into three parts
\begin{align}
\label{ec_err}
\nonumber
&\quad \big\|\ql{\ell-1}-\PE(H)\big\|_{\F}
\\[1ex]
\nonumber
& \leq \big\|\PE(H)-\ff(H)\big\|_{\F} +\big\|\PE(\Hlm)-\ff(\Hlm)\big\|_{\F} + \big\Vert{\ff(\widetilde{\Hlm})}-\ff(H)\big\Vert_{\F}
\\[1ex]
\nonumber
& =: I_1 + I_2 + I_3,
\end{align}
where the first two terms $I_1$ and $I_2$ correspond to the polynomial approximation error, and the last term $I_3$ captures the energy cutoff error. 

We first estimate $I_1$.
Note that $H$ is a scaled Hamiltonian such that all the eigenvalues lie in $[-1,1]$.
Therefore, we can obtain from \cref{cheberror} that
\begin{align}
\label{ec:I1}
I_1 = \big\|\PE(H)-\ff(H)\big\|_{\F} \leq C \big\|\PE-\ff\big\|_{L^\infty(\U)} \leq C e^{-\alpha M} .
\end{align}

By using a similar argument, we can estimate the second term $I_2$ as
\begin{equation}
\label{ec:I2}
I_2 \leq C e^{-\alpha M} .
\end{equation}

For the last term $I_3$, we decompose the basis into two ranges $|\G|^2\leq \frac{1}{8}\Ecl{\ell-1}$ and $\frac{1}{8}\Ecl{\ell-1}<|\G|^2<2\Ec$, and then obtain
\begin{align}
\label{ec:T1T2}
I_3
& \leq\Bigg(\sum_{\substack{|\G|^2\leq \frac{1}{8}\Ecl{\ell-1}\\|\G'|^2\leq \frac{1}{8}\Ecl{\ell-1}}}
+ 2\sum_{\substack{{\frac{1}{8}\Ecl{\ell-1}}<|\G|^2\leq 2\Ec\\|\G'|^2\leq {2\Ec}}}\Bigg)\left|{\ff(\widetilde{\Hlm})}_{\G\G'}-\ff(H)_{\G\G'}\right|
\\[1ex]
\nonumber
& =: T_1 + T_2.
\end{align}

To estimate $T_1$, we use \cref{cut_decay} from \Cref{app:lemma:trunc_converge}, which gives
\begin{align*}
\left|{\ff(\widetilde{\Hlm})}_{\G\G'}-\ff(H)_{\G\G'}\right|
& \leq Ce^{-\gamma{\sqrt{\Ecl{\ell-1}}}}
\leq Ce^{-\gamma\big(\sqrt{\Ecl{\ell-1}}-|\G|-|\G'|\big)},
\end{align*}
Therefore,
\begin{align}
\label{ec:T1}
T_1 \leq C\sum_{\substack{|\G|^2\leq \frac{1}{8}\Ecl{\ell-1}\\|\G'|^2\leq \frac{1}{8}\Ecl{\ell-1}}} e^{-\gamma\big(\sqrt{\Ecl{\ell-1}}-|\G|-|\G'|\big)}\leq Ce^{-\gamma\sqrt{\Ecl{\ell-1}}}.
\end{align}
We then use \cref{G_decay} from \Cref{app:lemma:G_decay} to obtain the estimate for $T_2$,
\begin{align}
\label{ec:T2}
T_2 & \leq 
\sum_{\substack{\frac{1}{8}\Ecl{\ell-1}<|\G|^2\leq 2\Ec\\|\G'|^2\leq 2\Ec,~|\G'|\leq|\G|}} e^{-\gamma\big(|\G'|+|\G-\G'|\big)}
+ \sum_{\substack{\frac{1}{8}\Ecl{\ell-1}<|\G|^2\leq 2\Ec\\|\G'|^2\leq 2\Ec,~|\G'|>|\G|}} e^{-\gamma\big(|\G|+|\G-\G'|\big)}
\\[1ex] 
\nonumber
& \leq Ce^{-\gamma\sqrt{\Ecl{\ell-1}}}.
\end{align}
This together with \cref{ec:T1T2} and \cref{ec:T1} implies
\begin{equation}
\label{ec:I3}
I_3\leq C e^{-\gamma\sqrt{\Ecl{\ell-1}}}.
\end{equation}

Combining the estimates of \eqref{ec:I1}, \eqref{ec:I2} and \eqref{ec:I3}, we can obtain \eqref{proof-A3-1}, which completes the proof of \Cref{thm:mlmc_ec}.
\hfill$\proofbox$

\subsection{Proof of \texorpdfstring{\Cref{thm:mlmc_pd}}{thm:mlmc:pd}}
\label{proof:thm:mlmc_pd}

Let $\ql{\ell-1}$ be given by \cref{spd_l}. 
Using \Cref{lemma:ml_var}, it suffices to prove that there exist constants $C$ and $\alpha$ depending on $\beta$, such that 
\begin{equation}
\label{proof-A4-1}
\big\|\ql{\ell-1}-\PE(H[\dm])\big\|_{\F}^2\leq Ce^{-2\alpha M^{(\ell-1)}}
\qquad \forall~1\leq\ell\leq L .
\end{equation}
Note that
\begin{equation}
\label{proof-A4-2}
\big\|\ql{\ell-1}-\PE(H[\dm])\big\|_{\F} \leq C\big\| p_{M^{(\ell-1)}} - p_{\PO} \big\|_{L^\infty(\U)} .
\end{equation}
By using the error estimate \cref{cheberror} for Chebyshev polynomial approximation, we have
\begin{align*}
\big\| p_{M^{(\ell-1)}}-p_{\PO} \big\|_{L^\infty(\U)} 
& \leq \big\| p_{M^{(\ell-1)}} -\ff \big\|_{L^\infty(\U)} + \big\| p_{\PO}-\ff \big\|_{L^\infty(\U)}
\\[1ex]
& \leq C \big(e^{-\alpha M^{(\ell-1)}} + e^{-\alpha \PO}\big) 
~\leq~ C e^{-\alpha M^{(\ell-1)}} .
\end{align*}
This together with \eqref{proof-A4-2} implies \eqref{proof-A4-1} and hence completes the proof.
\hfill$\proofbox$

% Bibliography goes here.
\bibliographystyle{siamplain}
\bibliography{references.bib}

\end{document}